\documentclass[a4paper,USenglish,cleveref,pdfa]{lipics-v2021}

\pdfoutput=1
\hideLIPIcs
\nolinenumbers

\usepackage{preamble}

\title{Parameterized Critical Node Cut Revisited}
\titlerunning{Parameterized Critical Node Cut Revisited}

\author{Du\v{s}an Knop}
{Department of Theoretical Computer Science, Faculty of Information Technology, Czech Technical University in Prague, Czech Republic
\and \url{https://knopdusa.pages.fit/}}
{dusan.knop@fit.cvut.cz}
{https://orcid.org/0000-0003-2588-5709}{}

\author{Nikolaos Melissinos}
{Computer Science Institute, Faculty of Mathematics and Physics, Charles University, Prague, Czech Republic}
{melissinos@iuuk.mff.cuni.cz}
{https://orcid.org/0000-0002-0864-9803}
{Partially supported by Charles University projects UNCE 24/SCI/008 and PRIMUS 24/SCI/012, and by the project 25-17221S of GAČR.}

\author{Manolis Vasilakis}
{Universit\'{e} Paris-Dauphine, PSL University, CNRS UMR7243, LAMSADE, Paris, France}
{emmanouil.vasilakis@dauphine.eu}
{https://orcid.org/0000-0001-6505-2977}
{Supported by the ANR project ANR-21-CE48-0022 (S-EX-AP-PE-AL) and
the Barrande Fellowship programme.}

\authorrunning{D. Knop, N. Melissinos, and M. Vasilakis}

\Copyright{Du\v{s}an Knop, Nikolaos Melissinos, and Manolis Vasilakis}



\category{} 

\relatedversiondetails[linktext={https://arxiv.org/abs/2506.23363}]{Full Version}{https://arxiv.org/abs/2506.23363}

\funding{}

\EventEditors{Neeldhara Misra and Magnus Wahlstr\"{o}m}
\EventNoEds{2}
\EventLongTitle{18th International Symposium on Parameterized and Exact Computation (IPEC 2023)}
\EventShortTitle{IPEC 2023}
\EventAcronym{IPEC}
\EventYear{2023}
\EventDate{September 6--8, 2023}
\EventLocation{Amsterdam, The Netherlands}
\EventLogo{}
\SeriesVolume{285}
\ArticleNo{22}

\begin{document}

\maketitle

\begin{abstract}
We study how to sparsify connectivity in graphs under a tight deletion budget.
Given a graph $G$ and integers $k,x \ge 0$,
\textsc{Critical Node Cut} (CNC) asks whether we can delete at most $k$ vertices so that
the number of remaining unordered pairs of connected vertices is at most $x$.
CNC generalizes \textsc{Vertex Cover} (the case $x=0$) and models tasks in
network design, epidemiology, and social network analysis.
We comprehensively map the structural parameterized complexity landscape for \textsc{Critical Node Cut}.
First, we prove W[1]-hardness for the combined parameter
$k + \mathrm{fes} + \Delta + \mathrm{pw}$,
where $\mathrm{fes}$ is the feedback edge set number,
$\Delta$ the maximum degree,
and $\mathrm{pw}$ the pathwidth of the input graph respectively.
This significantly improves over the known W[1]-hardness for $k+\mathrm{tw}$, where $\mathrm{tw}$ denotes the treewidth,
and is tight in that tree-depth together with maximum degree trivially yields FPT.
Second, we give new positive results.
Specifically, we identify three structural parameters--max-leaf number, vertex integrity, and modular-width--that
render the problem fixed-parameter tractable,
and develop a polynomial-time algorithm for graphs of constant clique-width.
Third, leveraging a technique introduced by Lampis~[ICALP '14], we develop an FPT approximation scheme that,
for any $\varepsilon > 0$,
computes a $(1+\varepsilon)$-approximate solution in time
$(\mathrm{tw} / \varepsilon)^{\mathcal{O}(\mathrm{tw})} n^{\mathcal{O}(1)}$,
where $\mathrm{tw}$ denotes the treewidth of the input graph.
Finally, we show that CNC admits no polynomial kernel when
parameterized by vertex cover number, unless standard assumptions fail.
Together, these results substantially sharpen the known complexity landscape for
CNC.

\end{abstract}


\section{Introduction}

Let $G$ be a simple, undirected, and unweighted graph, and let $k$ and $x$ be two non-negative integers.
The {\CNC} problem ({\CNCshort} for short) asks whether we can delete at most
$k$ vertices from $G$ so that the remaining graph contains at most $x$ connected vertex pairs.
This problem arises in various applications, such as budgeted immunization for contagion control,
analyzing brain connectivity, and uncovering structures in social graphs (see~\cite{cor/ArulselvanCEP09, tcs/HermelinKKN16} and references therein).

Importantly, {\CNCshort} generalizes the well-known \textsc{Vertex Cover} problem, which corresponds to the special case $x=0$.
Given its NP-hardness, we study {\CNCshort} through the lens of parameterized complexity.%
\footnote{We assume familiarity with the basics of parameterized complexity, as given e.g.~in~\cite{books/CyganFKLMPPS15}.}
Previous work has shown that the problem is fixed-parameter tractable (FPT) when parameterized by the vertex cover number
of the input graph~\cite{jgaa/SchestagGKS24}, but W[1]-hard when parameterized by the natural parameter $k$
or by the treewidth $\tw$~\cite{tcs/HermelinKKN16}.%
\footnote{More strongly, the reduction in~\cite{tcs/HermelinKKN16} implies W[1]-hardness with respect to tree-depth.}
As a matter of fact, Agrawal, Lokshtanov, and Mouawad have shown that this hardness persists even when parameterizing
by the combined parameter $k + \tw$~\cite{wg/AgrawalLM17}.

\subparagraph{Our Contribution.}
Our first result, which constitutes our main technical contribution,
is to significantly improve over the aforementioned hardness result of Agrawal, Lokshtanov, and Mouawad~\cite{wg/AgrawalLM17}.
In particular, in \cref{thm:CNC:fes} we prove that {\CNCshort} remains W[1]-hard even when parameterized by the combined parameter $k+\fes+\Delta+\pw$,
where $\fes$, $\Delta$, and $\pw$ denote the feedback edge set number, the maximum degree, and the pathwidth of the input graph respectively.
Notice that this also implies that the problem is W[1]-hard when parameterized by $k+\ctw$, where $\ctw$ denotes the cutwidth of the input graph,
as for any graph it holds that $\ctw \le \pw \cdot \Delta$.
In fact, this intractability result is tight in the sense that tree-depth together with maximum degree
trivially yields FPT, since this parameterization bounds the graph size.%
\footnote{A graph of tree-depth $\td$ and maximum degree $\Delta$ has at most $ \td \cdot \Delta^\td$ vertices.}
To obtain our result, we adapt a recent reduction for the related \textsc{Vertex Integrity} problem by Hanaka, Lampis, Vasilakis, and Yoshiwatari~\cite{mfcs/HanakaLVY24}.

Moving on, we obtain several results concerning the tractability of the problem under different structural parameterizations.
We identify three structural parameters that render {\CNCshort} fixed-parameter tractable,
thereby contrasting the hardness result above:
max-leaf number (\cref{thm:CNC:ml}),
vertex integrity (\cref{thm:CNC:vi}),
and modular-width (\cref{thm:CNC:mw}).
Our techniques towards obtaining those results include reductions to acyclic instances~\cite{mfcs/HanakaLVY24},
$N$-fold Integer Programming~\cite{KouteckyLO18},
and dynamic programming on suitable decompositions.
We also consider the parameterization by clique-width $\cw$, for which the problem is known to be W[1]-hard,
and obtain an algorithm with running time $n^{\bO(2^{\cw})}$ (\cref{thm:cw}) thus placing it in XP.

We then revisit the parameterization by treewidth and identify the bottleneck in standard dynamic programming as the need
to track large numerical values. To address this, we apply a rounding technique introduced by Lampis~\cite{icalp/Lampis14},
obtaining an FPT approximation scheme.
Specifically, in \cref{thm:fpt-as} we give an algorithm that, for any $\varepsilon > 0$,
computes in $(\tw / \varepsilon)^{\bO(\tw)} \cdot n^{\bO(1)}$ time a
deletion set of size at most $k$ such that the number of connected pairs in the resulting graph is at most
$(1+\varepsilon)$ times the minimum possible.

Finally, we investigate the compressibility of {\CNCshort}. We prove that, unless standard complexity assumptions fail,
the problem does not admit a polynomial kernel when parameterized by the vertex cover number of the input graph
(\cref{thm:kernel}).
This negative result highlights that even very restrictive structural parameterizations do not yield
efficient preprocessing algorithms (under standard assumptions).

\subparagraph{Related Work.}
{\CNCshort} is a well-studied problem; see~\cite[Section~5.2]{csr/LalouTK18} and the references therein.
Approximation algorithms have been proposed~\cite{cor/VentrescaA14},
while Hermelin, Kaspi, Komusiewicz, and Navon~\cite{tcs/HermelinKKN16} initiated the problem's study under parameterized complexity,
presenting a plethora of both positive and negative results with respect to natural and structural parameterizations.
Among others, they prove a (tight) $n^{o(k)}$ lower bound under the ETH, W[1]-hardness by tree-depth,
and that the problem is FPT parameterized by $x + \tw$.
Polynomial-time algorithms are known for trees~\cite{cor/SummaGL11},%
\footnote{That is the case even for weighted variants of the problem.}
and more generally,
the problem admits a DP algorithm of running time $n^{\bO(\tw)}$~\cite{dam/AddisSG13}.
As previously mentioned, {\CNCshort} is known to be W[1]-hard parameterized by $k+\tw$~\cite{wg/AgrawalLM17}.
On the other hand, it admits FPT algorithms with single-exponential parametric dependence when parameterized by vertex cover
number~\cite{jgaa/SchestagGKS24} or neighborhood diversity~\cite{mthesis/Schestag21}.

\section{Preliminaries}\label{sec:preliminaries}
Throughout the paper we use standard graph notations~\cite{Diestel17}
and assume familiarity with the basic notions of parameterized complexity~\cite{books/CyganFKLMPPS15}.
Proofs of statements marked with {\appsymbNote} are deferred to the appendix.

We use $\mathbb{N}$ to denote the set of non-negative integers.
For $x, y \in \mathbb{Z}$, let $[x, y] = \setdef{z \in \mathbb{Z}}{x \leq z \leq y}$,
while $[x] = [1,x]$.
For a set of integers $S \subseteq \mathbb{N}$
let $\binom{S}{c}$ denote the set of subsets of $S$ of size~$c$,
i.e., $\binom{S}{c} = \{S' \subseteq S : |S'| = c\}$.
For a (multi)set of positive integers $S$,
let $\Sigma(S)$ denote the sum of its elements.

All graphs considered are undirected without loops.
Given a graph $G$, $\cc(G)$ denotes the set of its connected components, while $\pairs(G)$ denotes the number of pairs of connected vertices in $G$,
that is, $\pairs(G) = \sum_{C \in \cc(G)} \binom{|V(C)|}{2}$.
For a vertex $v \in V(G)$, we denote the neighborhood of $v$ in $G$ by $N_{G}(v)$.
For a subset of vertices $S \subseteq V(G)$, $G[S]$ denotes the subgraph induced by $S$, while $G - S$ denotes $G[V(G) \setminus S]$.

For the basics on tree decompositions we refer to~\cite[Section~7]{books/CyganFKLMPPS15}.
We define the \emph{height of a node of a tree decomposition} as the distance from the root of the tree decomposition to said node.
In that case, the \emph{height of a tree decomposition} is defined as the maximum distance from the root to any of its nodes.

A graph of clique-width $k$ can be constructed through a sequence of the following operations on vertices that are labeled with at most $k$ different labels.
We can use (1) introducing a single vertex~$v$ of an arbitrary label~$i$, denoted $i(v)$,
(2) disjoint union of two labeled graphs, denoted $H_1 \oplus H_2$,
(3) introducing edges between \emph{all} pairs of vertices of two distinct labels~$i$ and~$j$ in a labeled graph~$H$, denoted $\eta_{i,j}(H)$,
and (4) changing the label of \emph{all} vertices of a given label~$i$ in a labeled graph~$H$ to a different label~$j$,
denoted $\rho_{i \to j}(H)$.
An expression describes a graph~$G$ if $G$ is the final graph given by the expression (after we remove all the labels).
The \emph{width} of an expression is the number of different labels it uses.
The clique-width of a graph is the minimum width of an expression describing it~\cite{dam/CourcelleO00}.
For a labeled graph $H$ and $v \in V(H)$, let $\lab_H(v)$ denote the label of $v$ in $H$,
while $\lab^{-1}_H (i) = \setdef{v \in V(H)}{\lab_H(v) = i}$ denotes the set of vertices of $H$ of label $i$.
A clique-width expression is \emph{irredundant} if whenever the operation $\eta_{i,j}$
is applied on a graph $G$, there is no edge between an $i$-vertex and a $j$-vertex in $G$,
and we will use such expressions to simplify our algorithms.
We remark that any clique-width expression can be transformed in linear time into an irredundant one of the same width~\cite{dam/CourcelleO00}.





\problemdef{\CNC}
{A graph $G = (V,E)$ as well as integers $k$ and $x$.}
{Determine whether there exists a set $S \subseteq V$ with $|S| \le k$
such that in $G-S$ there are at most $x$ pairs of connected vertices,
that is, whether $\pairs(G-S) \le x$.}

\begin{lemmarep}[\appsymb]\label{lemma:minimize_functions}
    Let $k \in \mathbb{N}$.
    Assume we are given $n$ functions $f_i \colon [0,k] \to \mathbb{N}$, where $i \in [n]$,
    and define $f \colon [0,k] \to \mathbb{N}$ such that for all $k' \in [0,k]$,
    \[
        f(k') = \min_{\genfrac{}{}{0pt}{}{k_i \in [0,k']}{k_1 + \ldots + k_n = k'}} \sum_{i \in [n]} f_i(k_i).
    \]
    One can compute the function $f$ in time $\bO(nk^2)$.
\end{lemmarep}

\begin{proof}
    We will perform dynamic programming.
    Consider the table $T[a][b]$, where $a \in [n]$ and $b \in [0,k]$.
    On a high-level, $T[a][b]$ contains the minimum sum over the first $a$ functions by
    partitioning a total budget of $b$ over them,
    that is,
    \begin{equation}\label{eq:minimize_functions}
        T[a][b] = \min_{\genfrac{}{}{0pt}{}{k_i \in [0,b]}{k_1 + \ldots + k_a = b}} \sum_{i \in [a]} f_i(k_i).
    \end{equation}
    In that case it suffices to compute the values of the table for $a=n$.

    We set $T[a][0] = \sum_{i \in [a]} f_i(0)$ for all $a \in [n]$,
    and $T[1][b] = f_1(b)$ for all $b \in [0,k]$.
    Then, for $a \in [2,n]$ and $b \in [k]$, let
    \[
        T[a][b] = \min_{ b' \in [0,b]} \left\{ T[a-1][b-b'] + f_a(b') \right\}.
    \]
    Notice that the DP table has a total of $n(k+1)$ cells,
    each of which requires $\bO(k)$ time to be filled, assuming that the evaluation of a function $f_i$ requires $\bO(1)$ time.
    Consequently, the total running time is $\bO(nk^2)$.
    As for the correctness of \Cref{eq:minimize_functions},
    one can easily argue about it by induction.
\end{proof}

\section{W[1]-hardness}\label{sec:w1_hardness}

In this section we prove that {\CNC} is W[1]-hard when parameterized by
the combined parameter $k+\fes+\Delta+\pw$, where $\fes$, $\Delta$, and $\pw$ denote the
feedback edge number, the maximum degree, and the pathwidth of the input graph, respectively.
This significantly improves over the previous result of Agrawal, Lokshtanov, and Mouawad~\cite{wg/AgrawalLM17}
who showed that the problem is W[1]-hard parameterized by $k$ plus the treewidth of the input graph.
Our reduction follows along the lines of a recent reduction by Hanaka et al.~\cite{mfcs/HanakaLVY24}
who showed the W[1]-hardness of {\SWCOC} under a similar parameterization.
We reduce from the {\RUBP} problem,
which we formally define below and which is known to be W[1]-hard parameterized by the number of bins $k$~\cite{mfcs/HanakaLVY24}.

\problemdef{\RUBP}
{A multiset $A = \braces{a_1, \ldots, a_n}$ of integers in unary, $k \in \mathbb{N}$, and a function $f \colon A \to \binom{[k]}{2}$.}
{Determine whether there is a partition of $A$ into multisets $\mathcal{A}_1, \ldots, \mathcal{A}_k$,
such that for all $i \in [k]$ it holds that
(i) $\Sigma(\mathcal{A}_i) = \Sigma(A) / k$,
and (ii) $\forall a \in \mathcal{A}_i$, $i \in f(a)$.}

We first present a sketch of our reduction.
The \emph{weight} of a vertex refers to the length of a path attached to it; this is later formalized in the construction.
Similarly to the construction of Hanaka et al.~\cite{mfcs/HanakaLVY24},
for every bin of the {\RUBP} problem instance we introduce a clique of $\bO(k^2)$ heavy vertices,
and then connect any pair of such cliques via two paths.
The weights are set in such a way that an optimal solution will delete only vertices from the connection paths.
In order to construct a path for a pair of bins, we compute the set of all subset sums of the items that can be placed in these two
bins, and introduce a vertex of medium weight per such subset sum.
This step can be completed in polynomial time, as all items are in unary.
Moreover, every such vertex corresponding to subset sum $s$ is preceded by exactly $s$ vertices of weight $1$.
An optimal solution will cut the path so that the number of vertices of weight $1$ will be partitioned between the two bins,
encoding the subset sum of the elements placed in each bin.
The second path that we introduce has balancing purposes, allowing us to exactly count the number of vertices of medium weight that every
connected component will end up with.
In the end, the intended solution deletes only vertices of medium weight,
resulting in components of only two possible sizes.

\begin{theorem}\label{thm:CNC:fes}
    {\CNC} is W[1]-hard parameterized by $k + \fes + \Delta + \pw$.
\end{theorem}

\begin{proof}
    Let $(A, k, f)$ be an instance of \RUBP,
    where $A = \braces{a_1, \ldots, a_n}$ denotes the multiset of items given in unary,
    $k$ is the number of bins, and
    $f \colon A \to \binom{[k]}{2}$ dictates the pair of bins an item may be placed into.
    In the following let
    $B = \Sigma(A) / k$,
    $c = 6k^2$,
    $M = kB+1$,
    $L = 36k^4 BM$,
    and $T = L + (k-1) \cdot 2B M + B$.
    Notice that $M > kB$ and $L / c = 6k^2BM$.
    We will reduce $(A,k,f)$ to an equivalent instance $(G,k',x)$ of \CNC,
    where $k' = 2 \binom{k}{2}$ and $x = k \cdot \binom{T}{2} + 2 \binom{k}{2} \cdot \binom{M-1}{2}$
    denote the size of the deletion set and the bound on the number of pairs of connected vertices, respectively.

    We proceed to describe the construction of graph $G$.
    For ease of presentation, when we say that a vertex $v \in V(G)$ has weight $w \ge 1$,
    it means that we introduce a path on $w-1$ vertices and connect one of its endpoints to $v$.

    For every $i \in [k]$, we introduce a clique on vertex set $\hat{C}_i$,
    which is comprised of $c$ vertices, each of weight $L / c$.
    Fix $i$ and $j$ such that $1 \leq i < j \leq k$, and let $H_{i,j} = \setdef{a \in A}{f(a) = \braces{i,j}}$ denote the
    multiset of all items of $A$ which can be placed either on bin $i$ or bin $j$, where $\Sigma(H_{i,j}) \leq 2B$.
    Let $\mathcal{S}(H_{i,j}) = \setdef{\Sigma(H)}{H \subseteq H_{i,j}}$ denote the set%
    \footnote{We stress that $\mathcal{S}(H_{i,j})$ is a set and not a multiset.}
    of all subset sums of $H_{i,j}$,
    and notice that since every element of $H_{i,j}$ is encoded in unary,
    $\mathcal{S}(H_{i,j})$ can be computed in polynomial time using,  e.g.,~Bellman's classical DP algorithm~\cite{Bellman}.
    We next construct two paths connecting the vertices of $\hat{C}_i$ and $\hat{C}_j$.
    First, we introduce the vertex set $\hat{U}_{i,j} = \setdef{v^{i,j}_q}{q \in [0, 4B - |\mathcal{S}(H_{i,j})| +1]}$,
    all the vertices of which are of weight $M$.
    Add edges $(v^{i,j}_q, v^{i,j}_{q+1})$ for all $q \in [0, 4B - |\mathcal{S}(H_{i,j})|]$,
    as well as $(v_1, v^{i,j}_0)$ and $(v^{i,j}_{4B - |\mathcal{S}(H_{i,j})| +1}, v_2)$,
    for all $v_1 \in \hat{C}_i$ and $v_2 \in \hat{C}_j$.
    Next, we introduce the vertex set
    $\hat{D}_{i,j} = \setdef{s^{i,j}_q}{q \in [\Sigma(H_{i,j})]} \cup \setdef{\sigma^{i,j}_q}{q \in \mathcal{S}(H_{i,j})}$,
    with the $s$-vertices being of weight $1$ and the $\sigma$-vertices of weight $M$.
    Then, add the following edges:
    \begin{itemize}
        \item $(v_1, \sigma^{i,j}_0)$ and $(\sigma^{i,j}_{\Sigma(H_{i,j})}, v_2)$, for all $v_1 \in \hat{C}_i$ and $v_2 \in \hat{C}_j$,

        \item for $q \in \mathcal{S}(H_{i,j})$,
        add the edges $(s^{i,j}_q, \sigma^{i,j}_q)$ if $q \neq 0$
        and $(\sigma^{i,j}_q, s^{i,j}_{q+1})$ if $q \neq \Sigma(H_{i,j})$, and

        \item for $q \in [\Sigma(H_{i,j})] \setminus \mathcal{S}(H_{i,j})$,
        add the edge $(s^{i,j}_q, s^{i,j}_{q+1})$.
    \end{itemize}

    This concludes the construction of $G$.
    See \Cref{fig:fes_construction} for an illustration.
    Notice that the number of vertices of $G$ excluding those belonging to the cliques $\hat{C}_1, \ldots, \hat{C}_k$ or to a path attached to them
    is $(4B+2)M \cdot \binom{k}{2} + kB < 6k^2BM = L/c$.

    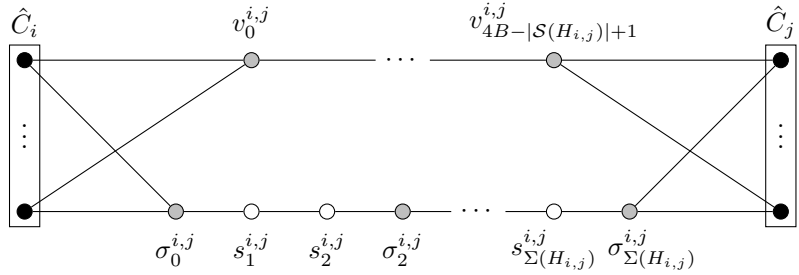
\begin{figure}[ht]
        \centering

        \begin{tikzpicture}[scale=1, transform shape]

        \node[black_vertex] (c11) at (5,5) {};
        \node[] () at (5,6.1) {$\vdots$};
        \node[black_vertex] (c1n) at (5,7) {};
        \draw[] (4.8,4.8) rectangle (5.2,7.2);
        \node[] () at (5,7.5) {$\hat{C}_i$};

        \begin{scope}[shift={(10,0)}]
            \node[black_vertex] (c21) at (5,5) {};
            \node[] () at (5,6.1) {$\vdots$};
            \node[black_vertex] (c2n) at (5,7) {};
            \draw[] (4.8,4.8) rectangle (5.2,7.2);
            \node[] () at (5,7.5) {$\hat{C}_j$};
        \end{scope}

        \begin{scope}[shift={(3,-3)}]
            \node[gray_vertex] (p0) at (5,10) {};
            \node[] () at (5,10.5) {$v^{i,j}_0$};
            \node[] (temp) at (7,10) {$\cdots$};
            \node[gray_vertex] (pB) at (9,10) {};
            \node[] () at (9,10.5) {$v^{i,j}_{4B - |\mathcal{S}(H_{i,j})| + 1}$};
            \draw[] (p0)--(temp)--(pB);
            \draw[] (p0)--(c11);
            \draw[] (p0)--(c1n);
            \draw[] (pB)--(c21);
            \draw[] (pB)--(c2n);
        \end{scope}

        \begin{scope}[shift={(2,0)}]
            \node[gray_vertex] (sigma0) at (5,5) {};
            \node[] () at (5,4.5) {$\sigma^{i,j}_0$};
            \node[vertex] (s1) at (6,5) {};
            \node[] () at (6,4.5) {$s^{i,j}_1$};
            \node[vertex] (s2) at (7,5) {};
            \node[] () at (7,4.5) {$s^{i,j}_2$};
            \node[gray_vertex] (sigma2) at (8,5) {};
            \node[] () at (8,4.5) {$\sigma^{i,j}_2$};
            \node[] (temp2) at (9,5) {$\cdots$};
            \node[vertex] (sB) at (10,5) {};
            \node[] () at (10,4.5) {$s^{i,j}_{\Sigma(H_{i,j})}$};
            \node[gray_vertex] (sigmaB) at (11,5) {};
            \node[] () at (11.3,4.5) {$\sigma^{i,j}_{\Sigma(H_{i,j})}$};

            \draw[] (sigma0)--(s1)--(s2)--(sigma2)--(temp2)--(sB)--(sigmaB);
            \draw[] (sigma0)--(c11);
            \draw[] (sigma0)--(c1n);
            \draw[] (sigmaB)--(c21);
            \draw[] (sigmaB)--(c2n);
        \end{scope}

        \end{tikzpicture}

        \caption{Rectangles denote cliques of size $c$.
        Here we assume that $1 \leq i < j \leq k$,
        $1 \notin \mathcal{S}(H_{i,j})$, and $2 \in \mathcal{S}(H_{i,j})$.
        The gray and black color indicates weight of $M$ and $L / c$, respectively.}
        \label{fig:fes_construction}
    \end{figure}

    \begin{lemmarep}[\appsymb]\label{lemma:rubp->cnc}
        If $(A,k,f)$ is a Yes-instance of \RUBP, then $(G,k',x)$ is a Yes-instance of \CNC.
    \end{lemmarep}

    \begin{proof}
        Let $(\mathcal{A}_1, \ldots, \mathcal{A}_k)$ be a partition of $A$ such that for all $i \in [k]$ it holds that
        (i) $\Sigma(\mathcal{A}_i) = B$, and (ii) $\forall a \in \mathcal{A}_i$, $i \in f(a)$.
        Fix $1 \leq i < j \leq k$, and let $\mathcal{H}_{i,j} \subseteq H_{i,j}$ such that $\mathcal{A}_i \cap H_{i,j} = \mathcal{H}_{i,j}$,
        while $\mathcal{A}_j \cap H_{i,j} = H_{i,j} \setminus \mathcal{H}_{i,j}$.
        Note that for all $i \in [k]$, it holds that
        \begin{equation}\label{eq:rubp->cnc}
            \sum_{j \in [i-1]} \big( \Sigma(H_{j,i}) - \Sigma(\mathcal{H}_{j,i}) \big) + \sum_{j \in [i+1,k]} \Sigma(\mathcal{H}_{i,j}) = B.
        \end{equation}
        Additionally, let $r_{i,j} = |\mathcal{S}(H_{i,j}) \cap [0,\Sigma(\mathcal{H}_{i,j})-1]|$ denote
        the number of distinct sums of value at most $\Sigma(\mathcal{H}_{i,j})-1$ obtained from subsets of $H_{i,j}$.

        We are now ready to construct the cut~$S$.
        Set $S = \setdef{\sigma^{i,j}_{\Sigma(\mathcal{H}_{i,j})}, v^{i,j}_{2B - r_{i,j}}}{1 \leq i < j \leq k}$.
        It holds that $|S| = 2 \binom{k}{2} = k'$.
        Additionally, it holds that $G-S$ has exactly $k + 2 \binom{k}{2}$ connected components:
        every clique $\hat{C}_i$ belongs to a distinct component,
        while all vertices of $S$ are of weight $M$,
        thus resulting in $|S|$ additional components of size $M-1$.
        To prove the statement, it suffices to show that the connected components of $G-S$ apart from the latter $|S|$ ones have size at most $T$.

        Let for all $i \in [k]$, $\mathcal{C}_i$ denote the connected component of $G-S$ that contains the vertices of $\hat{C}_i$.
        Fix $1 \leq i < j \leq k$.
        Since $S \cap (\hat{D}_{i,j} \cup \hat{U}_{i,j}) = \{ \sigma^{i,j}_{\Sigma(\mathcal{H}_{i,j})}, v^{i,j}_{2B - r_{i,j}} \}$,
        it holds that $\mathcal{C}_i$ contains $\Sigma(\mathcal{H}_{i,j})$ vertices of weight $1$ belonging to $\hat{D}_{i,j}$,
        as well as $r_{i,j} + (2B - r_{i,j}) = 2B$ vertices of weight $M$ from $\hat{D}_{i,j} \cup \hat{U}_{i,j}$.
        As for $\mathcal{C}_j$, it contains $\Sigma(H_{i,j}) - \Sigma(\mathcal{H}_{i,j})$ vertices of weight $1$ belonging to $\hat{D}_{i,j}$,
        as well as $(|\mathcal{S}(H_{i,j})| - 1 - r_{i,j}) + (4B - |\mathcal{S}(H_{i,j})| + 2 - 1 - (2B - r_{i,j})) = 2B$ vertices of weight $M$ from $\hat{D}_{i,j} \cup \hat{U}_{i,j}$.
        For any fixed $i \in [k]$, it follows that
        \begin{align*}
            |V(\mathcal{C}_i)| &= L +
            \sum_{j \in [i-1]} \big( \Sigma(H_{i,j}) - \Sigma(\mathcal{H}_{i,j}) + 2BM \big) +
            \sum_{j \in [i+1,k]} \big( \Sigma(\mathcal{H}_{i,j}) + 2BM \big)\\
            &= L + (k-1) \cdot 2BM + B\\
            &= T,
        \end{align*}
        where the second equality is due to \Cref{eq:rubp->cnc}.
        This concludes the proof.
    \end{proof}

    \begin{lemma}\label{lemma:cnc->rubp}
        If $(G,k',x)$ is a Yes-instance of \CNC, then $(A,k,f)$ is a Yes-instance of \RUBP.
    \end{lemma}

    \begin{proof}
        Let $S \subseteq V(G)$ with $|S| \leq k' = 2 \binom{k}{2}$
        such that the number of pairs of connected vertices in $G - S$ is minimized and is at most $x$.
        Notice that $|S| = k'$, as otherwise there always exists a vertex whose deletion reduces said number.
        For all $i \in [k]$ it holds that $|\hat{C}_i| > k'$,
        thus the vertices of $\hat{C}_i \setminus S$ belong to the same connected component of $G-S$.

        \begin{claim}\label{claim:cnc->rubp:del_per_path}
            For any $i,j \in [k]$ with $1 \le i < j \le k$ it holds that
            $|S \cap \hat{U}_{i,j}| \le 1$ and $|S \cap \hat{D}_{i,j}| \le 1$.
        \end{claim}

        \begin{claimproof}
            Let $i,j \in [k]$ with $1 \le i < j \le k$ such that $|S \cap \hat{Z}_{i,j}| \ge 2$,
            where $\hat{Z}_{i,j} \in \{ \hat{U}_{i,j}, \hat{D}_{i,j} \}$,
            and let $v$ denote the rightmost vertex of $\hat{Z}_{i,j}$ that belongs to $S$,
            that is, no vertex of $\hat{Z}_{i,j}$ between $v$ and the vertices of $\hat{C}_j$ belongs to $S$.
            Let $\mathcal{C}_j$ denote the connected component of $G-S$ that contains the vertices of $\hat{C}_j$.
            It holds that $|V(\mathcal{C}_j)| \ge 2L/3$, since $c \ge 3k'$.
            Let $v' \in \hat{C}_j \setminus S$ and consider the deletion set $S' = (S \setminus \{ v \}) \cup \{ v' \}$,
            where $|S'| = |S|$.

            We argue that $\pairs(G-S') < \pairs(G-S)$, thus contradicting the optimality of $S$.
            Indeed, starting from $G-S$ it holds that not deleting $v$ increases the number of connected pairs by at most
            $4BM \cdot |V(\mathcal{C}_j)| + 4BM \cdot M = (2L/3 k^2 c) \cdot (|V(\mathcal{C}_j)| + M)$,
            while subsequently deleting $v'$ decreases the number of pairs by at least
            $(|V(\mathcal{C}_j)| - L/c) \cdot (L/c)$.
            It holds that
            \begin{align*}
                (|V(\mathcal{C}_j)| - L/c) \cdot (L/c) > (2L/3 k^2 c) \cdot (|V(\mathcal{C}_j)| + M) \iff\\
                |V(\mathcal{C}_j)| - L/c > (2/3 k^2) \cdot (|V(\mathcal{C}_j)| + M) \iff\\
                |V(\mathcal{C}_j)| \left( 1 - \frac{2}{3k^2} \right) - L/c > \frac{2}{3} k^2 M \impliedby\\
                |V(\mathcal{C}_j)| \left( 1 - \frac{2}{3k^2} \right) - L/c > L/c \impliedby\\
                \frac{2L}{3} \cdot \left( 1 - \frac{2}{3k^2} \right) > \frac{2L}{6k^2} \iff\\
                \left( 1 - \frac{2}{3k^2} \right) > \frac{1}{2k^2} \iff\\
                1 \ge \frac{5}{6k^2},
            \end{align*}
            which holds for all $k \ge 1$ and
            where we used the fact that $\frac{2}{3} k^2 M < \frac{L}{c}$ as well as that $|V(\mathcal{C}_j)| \ge 2L/3$.
        \end{claimproof}

        We next argue that every connected component of $G-S$ contains vertices belonging to at most one clique $\hat{C}_i$.

        \begin{claim}\label{claim:cnc->rubp:main_claim}
            For any connected component $\mathcal{C}$ of $G-S$ it holds that
            \[ |\setdef{i \in [k]}{V(\mathcal{C}) \cap \hat{C}_i \neq \varnothing}| \le 1 \,. \]
        \end{claim}

        \begin{claimproof}
            We first fix some notation used throughout the proof.
            Let $\Pi \subseteq 2^{[k]}$ be the partition of $[k]$ such that
            $p \in \Pi$ if and only if there exists a connected component $\mathcal{C}_p \in \cc (G-S)$
            with $p = \setdef{i \in [k]}{V(\mathcal{C}_p) \cap \hat{C}_i \neq \varnothing}$.
            Notice that for $ p \neq \varnothing$, such a component is unique.
            For $p \in \Pi$ we further define the set $\hat{R}(p)$
            which is composed of the vertices $\bigcup_{i \in p} \hat{C}_i \cup \bigcup_{i,j \in p} (\hat{U}_{i,j} \cup \hat{D}_{i,j})$
            plus the vertices belonging to the paths attached to the former vertices.
            Let $g \colon \cc(G-S) \to [0,k]$ such that
            $g(\mathcal{C}) = |\setdef{i \in [k]}{V(\mathcal{C}) \cap \hat{C}_i \neq \varnothing}|$ for all $\mathcal{C} \in \cc(G-S)$.
            Consider a component $\mathcal{C}_{p_{\max}} \in \cc(G-S)$ that maximizes $g$, where $g(\mathcal{C}_{p_{\max}}) = |p_{\max}| \ge 1$.
            It suffices to prove that $|p_{\max}| < 2$.
            Towards a contradiction assume that $|p_{\max}| \ge 2$.

            Let $\mathcal{C}_{p'}$ with $g(\mathcal{C}_{p'}) = |p'| < |p_{\max}|$.
            We argue that $S \cap \hat{R}(p') = \varnothing$.
            To see this, notice that the deletion of such a vertex results in at most
            \[
                \frac{L}{c} \cdot \left( |p'| \cdot L + \frac{L}{c} \right)
            \]
            disconnected pairs, where we use the fact that the number of vertices of $G$
            excluding those belonging to the cliques $\hat{C}_1, \ldots, \hat{C}_k$
            or to a path attached to them is less than $L/c$.
            On the other hand, $|V(\mathcal{C}_{p_{\max}})| \ge |p_{\max}| \cdot L - k' \cdot (L/c) \ge (|p_{\max}|-1) \cdot L + (2L)/3$,
            thus the deletion of a vertex $v \in \hat{C}_i \setminus S$ with $i \in p_{\max}$ results in at least
            \[
                \frac{L}{c} \cdot \left( (|p_{\max}|-1) \cdot L + \frac{2L}{3} \right) > \frac{L}{c} \cdot \left( |p'| \cdot L + \frac{L}{c} \right)
            \]
            disconnected pairs.

            The previous paragraph, along with \Cref{claim:cnc->rubp:del_per_path} and the fact that $|S| = 2\binom{k}{2}$,
            imply that there exists $p \in \Pi$ with $|p| = |p_{\max}|$ such that
            $|S \cap \hat{R}(p)| \ge 2 \binom{|p|}{2}$.
            Let $S' = (S \setminus \hat{R}(p)) \cup \setdef{v^{i,j}_0, \sigma^{i,j}_0}{i,j \in p}$,
            where $|S'| \le |S|$.
            We argue that $\pairs(G-S') < \pairs(G-S)$.
            Let $x_1$ (resp., $x_2$) denote the number of pairs of connected vertices in $G-S$ (resp., $G-S'$)
            such that at least one vertex belongs to $\hat{R}(p)$ (notice that the rest of the pairs remain the same in the two graphs).
            It holds that
            \begin{align*}
                2 \cdot x_1 &> 2 \cdot \binom{(|p|-1) \cdot L + 2L/3}{2}\\
                &\ge ((|p|-1) L + 2L/3) \cdot ((|p|-1) L + L/3)\\
                &= L^2 (|p|^2 - |p| + 2/9),
            \end{align*}
            where the first inequality is due to the component of $G-S$ that contains the vertices of $\setdef{\hat{C}_i}{i \in p} \setminus S$.
            On the other hand, in $G-S'$ the vertices of each $\hat{C}_i$ for $i \in p$ are in distinct connected components,
            thus it holds that
            \begin{align*}
                2 \cdot x_2 &< 2 |p| \cdot \binom{L + L/c}{2}\\
                &\le |p| (L + L/c)^2\\
                &= L^2 \big( |p| + |p| \cdot \frac{2c+1}{c^2} \big)\\
                &\le L^2 (|p| + 2/9),
            \end{align*}
            where the first inequality is due to the fact that the number of vertices of $G$ apart from those of $\bigcup_{i \in [k]} \hat{C}_i$ or attached to them is at most $L/c$.
            The last inequality is due to the fact that
            $|p| \cdot \frac{2c+1}{c^2} \le k \cdot \frac{12k^2+1}{36k^4}$,
            which for $k \ge 2$ is upper-bounded by $2/9$.

            Finally, we have that
            \begin{align*}
                x_1 &> x_2 \impliedby\\
                L^2 (|p|^2 - |p| + 2/9) &\ge L^2 (|p| + 2/9) \iff\\
                |p|^2 &\ge 2|p|,
            \end{align*}
            which holds as by assumption $|p| \ge 2$.
            This contradicts the optimality of set $S$.
        \end{claimproof}

        By \Cref{claim:cnc->rubp:main_claim} and the fact that $k' = 2\binom{k}{2}$ it follows that
        $S$ contains exactly one vertex per path between cliques.
        In that case, since the vertices of those paths are of weight $1$ or $M$,
        regarding the connected components of $G-S$ it holds that there are exactly $k$ of size at least $L$,
        as well as $a \in [\binom{k}{2}, \, 2\binom{k}{2}]$ of size exactly $M-1$,
        where $a \ge \binom{k}{2}$ since half the paths between cliques contain only vertices of weight $M$.
        Further notice that the number of pairs of connected vertices is minimized when each component of size at least $L$
        contains $\frac{|V(G)| - a \cdot M}{k}$ vertices,%
        \footnote{To see this notice that $\binom{n - \varepsilon}{2} + \binom{n + \varepsilon}{2} > 2 \cdot \binom{n}{2}$ for all $\varepsilon > 0$.}
        while for all $a \le 2\binom{k}{2}$ it holds that
        \[
            k \cdot \binom{(|V(G)| - a \cdot M) / k}{2} + a \cdot \binom{M-1}{2} >
            k \cdot \binom{(|V(G)| - (a+1) \cdot M) / k}{2} + (a+1) \cdot \binom{M-1}{2}.
        \]
        Consequently, the minimum number of pairs of connected vertices is obtained exactly when there are
        $2 \binom{k}{2}$ components in $G-S$ of size $M-1$, and the rest of the vertices are partitioned equally among the remaining $k$ components.
        Notice that
        \begin{align*}
            |V(G)| &= k \cdot L + (4B+2) M \cdot \binom{k}{2} + kB\\
            &= k \cdot (L + 2BM(k-1) + B) + 2M \cdot \binom{k}{2}\\
            &= k \cdot T + 2 \binom{k}{2} \cdot M,
        \end{align*}
        thus since $x = k \cdot \binom{T}{2} + 2 \binom{k}{2} \cdot \binom{M-1}{2}$ it follows that this is indeed the case for $G-S$.
        Consequently, for all $1 \le i < j \le k$, there exists $\sigma^{i,j}_q \in S$ with $q \in \mathcal{S}(H_{i,j})$.

        Let $\mathcal{C}_i$ denote the connected component of $G - S$ containing the vertices of $\hat{C}_i$.
        Since $\mathcal{C}_i$ contains $T$ vertices,
        while $kB < M$, it follows that $\hat{\mathcal{C}}_i$ contains exactly $B$ vertices of weight~$1$.
        Let $(\mathcal{A}_1, \ldots, \mathcal{A}_k)$ be a partition of $A$ defined in the following way:
        for all $1 \leq i < j \leq k$, if $\sigma^{i,j}_q \in S$,
        then $\mathcal{A}_i \cap H_{i,j} = \mathcal{H}_{i,j}$ and $\mathcal{A}_j \cap H_{i,j} = H_{i,j} \setminus \mathcal{H}_{i,j}$,
        where $\mathcal{H}_{i,j} \subseteq H_{i,j}$ such that $\Sigma(\mathcal{H}_{i,j}) = q$.
        Notice that $\Sigma(\mathcal{A}_i)$ is equal to the number of vertices of weight $1$ in $\mathcal{C}_i$,
        therefore $\Sigma(\mathcal{A}_i) = B$ follows.
    \end{proof}

    \begin{lemma}\label{lemma:fes:parameter_bound}
        It holds that $\fes(G) = \bO (k^5)$,
        $\Delta(G) = \bO(k^2)$,
        and $\pw(G) = \bO(k^3)$.
    \end{lemma}
    \begin{proof}
        For the feedback edge number, let $F \subseteq E(G)$ contain all edges between vertices of~$\hat{C}_i$, for all $i \in [k]$,
        as well as all edges between such vertices and endpoints of paths~$\hat{U}_{i,j}$ and~$\hat{D}_{i,j}$.
        Notice that
        \[
            |F| = k \cdot \binom{c}{2} + 4\binom{k}{2} \cdot c
            = \bO(k^5),
        \]
        while the graph remaining after the deletion of the edges in $F$ is a caterpillar forest.
        For the pathwidth bound notice that $G - \bigcup_{i=1}^k \hat{C}_i$ is a caterpillar forest,
        thus $\pw(G) = \bO(k^3)$.
        Lastly, for the maximum degree notice that for all \mbox{$v \in \bigcup_{i \in [k]} \hat{C}_i$,
        $|N(v)| = c + 2(k-1) = \bO(k^2)$}.
        As for the vertices of $\hat{D}_{i,j} \cup \hat{U}_{i,j}$,
        they are of degree at most $3$,
        apart from the endpoints which are neighbors with all the vertices of a single clique,
        therefore of degree~$\bO(k^2)$.
        Any remaining vertex is of degree at most $2$.
    \end{proof}
    Due to \Cref{lemma:rubp->cnc,lemma:cnc->rubp,lemma:fes:parameter_bound}, \Cref{thm:CNC:fes} follows.
\end{proof}

\section{Algorithms}\label{sec:algo}

In this section we identify several structural parameters that render {\CNC} fixed-parameter tractable,
namely max-leaf number, vertex integrity, and modular-width.
Furthermore, we show that the problem parameterized by clique-width belongs to XP.

\subsection{Max-Leaf Number}

Here we present a single-exponential algorithm for {\CNC} parameterized by the max-leaf number of the input graph.
This is a well-studied but very restricted parameter~\cite{ejc/FellowsJR13,mst/FellowsLMMRS09,algorithmica/Lampis12}.
Note that this nicely complements the W[1]-hardness of the problem parameterized by $\fes+\Delta$ which follows by \Cref{thm:CNC:fes};
graphs of max-leaf number at most $k$ are known to be a subdivision of a graph of at most $\bO(k)$ vertices~\cite{siamdm/KleitmanW91},
thus their feedback edge set number and their maximum degree are bounded by a function of $k$.
As a matter of fact, the reduction of \Cref{thm:CNC:fes} produces a graph that, on a first look, appears to have bounded
max-leaf number. The value of this parameter increases due to attaching long paths on the weighted vertices, and \Cref{thm:CNC:ml}
implies that this increase is unavoidable, as the problem is FPT under this parameterization.

Our algorithm closely follows the one presented by Hanaka, Lampis, Vasilakis, and Yoshiwatari~\cite{mfcs/HanakaLVY24}
for {\VI} parameterized by max-leaf number.
In particular, our first step is to contract any cycles, reducing the instance to a weighted forest on which some of the vertices
are marked as undeletable.
Then we apply a simple DP algorithm for trees in order to compute, for every connected component of the instance,
the minimum number of pairs of connected vertices remaining after at most $b$ vertex deletions from said component, for all $b \in [0,k]$.
Having computed this for every component, applying \Cref{lemma:minimize_functions} then allows us to determine the
number of deletions per component such that the total number of connected pairs is minimized.

\begin{theoremrep}[\appsymb]\label{thm:CNC:ml}
    There is an algorithm that given any instance $\mathcal{I} = (G,k,x)$ of \CNC,
    it decides $\mathcal{I}$ in time $2^{\bO(\ml)} n^{\bO(1)}$,
    where $\ml$ denotes the max-leaf number of $G$.
\end{theoremrep}

\begin{proof}
    Let $X = \setdef{v \in V}{\deg_G(v) \ge 3}$ denote the set of vertices of degree at least $3$ in $G$.
    By a well-known result by Kleitman and West~\cite{siamdm/KleitmanW91} it holds that $|X| = \bO(\ml)$.
    We make use of the following result of Hanaka, Lampis, Vasilakis, and Yoshiwatari~\cite{mfcs/HanakaLVY24}.

    \begin{lemma}[{\cite[Lemma 9]{mfcs/HanakaLVY24}}]
        Let $G=(V,E)$ be a graph and $X$ be a set of vertices such that all vertices
        of $V \setminus X$ have degree at most $2$ in $G$.
        For all positive integers $\ell,p$, if there exists a separator $S$ of size at most $p$
        such that all components of $G - S$ have size at most $\ell$,
        then there exists such a separator $S$ that also satisfies the following property:
        for every cycle $C$ of $G$ with $C \cap X \neq \varnothing$ we
        either have $C \cap S = \varnothing$ or $C \cap X \cap S\neq \varnothing$.
    \end{lemma}

    The proof of~\cite[Lemma 9]{mfcs/HanakaLVY24} uses an exchange argument such that
    for every component either its size remains the same or it is split into multiple components.
    In that case, exactly the same proof can be used to show the following.

    \begin{observation}\label{observation:ml:good_seperator}
        Let $G=(V,E)$ be a graph and $X$ be the set of its vertices of degree at least $3$ in $G$.
        For all positive integers $k,x$, if there exists a separator $S$ of size at most $k$
        such that the number of pairs of connected vertices in $G - S$ is at most $x$,
        then there exists such a separator $S$ that also satisfies the following property:
        for every cycle $C$ of $G$ with $C \cap X \neq \varnothing$ we
        either have $C \cap S = \varnothing$ or $C \cap X \cap S \neq \varnothing$.
    \end{observation}

    Now, fix an optimal solution $S \subseteq V$ of size $|S| \le k$ that maximizes $|S \cap X|$.
    We first guess its intersection with $X$ and delete these vertices,
    while marking the vertices of $X \setminus S$ undeletable.
    For every component $\mathcal{C}_i \in \cc(G-(S \cap X))$,
    we compute a function \mbox{$f_i \colon [0,k - |S \cap X|] \to \left[ \binom{|V(\mathcal{C}_i)|}{2} \right]$},
    such that for all $j \in [0,k - |S \cap X|]$,
    $f_i(j)$ is equal to the minimum number of pairs of connected vertices of $\mathcal{C}_i$
    after the deletion of exactly $j$ of its vertices.
    Notice that if $V(\mathcal{C}_i) \cap X = \varnothing$,
    then $\mathcal{C}_i$ is either a path or a cycle and it is straightforward to compute $f_i$.

    Assume that $V(\mathcal{C}_i) \cap X \neq \varnothing$.
    If $\mathcal{C}_i$ does not contain any cycles, then it is a tree.
    Otherwise, let $C$ denote a cycle appearing in $\mathcal{C}_i$, in which case it holds that $C \cap X \neq \varnothing$,
    and by \Cref{observation:ml:good_seperator} it follows that $S \cap C = \varnothing$.
    In that case, we contract $C$ to a single vertex $v$ of weight equal to the sum of weights of the vertices of the cycle,
    that is, $\wt(v) = \sum_{c \in C} \wt(c)$,
    mark it as undeletable and add it to $X$.
    We repeat this exhaustively, resulting in a component which is a weighted tree.
    In that case one can easily adapt the DP algorithm for trees by Di Summa, Grosso, and Locatell~\cite{cor/SummaGL11}
    to account for the non-deletability and the weights and compute $f_i$ in polynomial time.

    Finally, having computed for every component $\mathcal{C}_i \in \cc(G - (S \cap X))$ the function $f_i$,
    it suffices to use \Cref{lemma:minimize_functions}
    in order to compute a $k_i \in [0,k - |S \cap X|]$ for every $\mathcal{C}_i$
    so that $\sum k_i \le k$ and $\sum f_i(k_i)$ is minimized,
    and verify whether $\sum f_i(k_i) \le x$.
    Regarding the running time of our algorithm, guessing $S \cap X$ incurs an overhead of $2^{\bO(\ml)}$,
    while the rest of the steps require polynomial time.
\end{proof}

\subsection{Vertex Integrity}

The \emph{vertex integrity} of a graph $G = (V,E)$, denoted $\vi(G)$, is the minimum integer $k$
such that there is a vertex set $U \subseteq V$ with $|U| + \max_{C \in \cc(G-S)} |V(C)| \le k$,
where $\cc(G-S)$ denotes the set of connected components of $G-S$.
Here we show that the parameterization by vertex integrity renders {\CNC} fixed-parameter tractable.
Given the fact that the vertex integrity of a graph is upper-bounded by its vertex cover number,
this improves over the known FPT algorithm parameterized by vertex cover number~\cite{jgaa/SchestagGKS24}.
On the other hand, the vertex integrity of a graph is lower-bounded by its tree-depth,
under which parameterization the problem remains W[1]-hard~\cite{tcs/HermelinKKN16},
so our result is in that sense tight.

\begin{theoremrep}[\appsymb]\label{thm:CNC:vi}
    There is an algorithm that given any instance $\mathcal{I} = (G,k,x)$ of \CNC,
    it decides $\mathcal{I}$ in time ${\vi}^{\bO(\vi^2)} n^{\bO(1)}$,
    where $\vi$ denotes the vertex integrity of $G$.
\end{theoremrep}

\begin{proof}
    Let $U \subseteq V(G)$ such that for all connected components $C \in \cc(G-S)$ it holds that
    $|U| + |V(C)| \le \vi(G)$; such a set can be computed in time ${\vi}^{\bO(\vi)} n^{\bO(1)}$~\cite{algorithmica/DrangeDH16}.
    Furthermore, let $\cc(G-U) = \{ C^1, \ldots, C^m \}$ denote the connected components of $G-U$, where $m \in [0,n]$;
    by assumption, it holds that $|V(C^j)| \le \vi$ for all $j \in [m]$.

    On a high level, our algorithm consists of two separate steps.
    In Step 1, we guess the intersection of $U$ with an optimal deletion set, as well as the way the non-deleted vertices of $U$
    are partitioned in connected components in the final graph.
    All remaining vertices of the optimal deletion set belong to components of $G - U$.
    Next, it suffices to encode the choices in the rest of the graph as an $m$-fold IP (see~\cite{KouteckyLO18}) of a small number of global constraints
    and small coefficients in the right-hand side, and find a solution that indeed minimizes the number of pairs of connected
    vertices in the final graph, while respecting the guesses of Step 1.

    Fix a deletion set $S^* \subseteq V(G)$ such that $|S^*| \le k$ and the number of pairs of connected vertices in $G-S^*$ is minimum.
    We start by guessing the intersection of said deletion set with the modulator $U$, that is, let $S = U \cap S^*$.
    We further guess how $U \setminus S$ is partitioned into connected components in $G - S^*$,
    and let $V_1, \ldots, V_\ell$ denote this partition, where $\ell \le \vi$.
    Notice that for $i \in [\ell]$, $G[V_i]$ is not necessarily connected,
    however it holds that if $uv \in E$ and $u,v \in U \setminus S$, then $u$ and $v$ must belong to the same connected component of $G-S^*$.
    Any guess $S, V_1, \ldots, V_\ell$ that adheres to the latter constraint, that is,
    no connected component of $G[U \setminus S]$ intersects more than one blocks of the partition, will be called \emph{valid}.

    For each valid guess, we will create an $m$-fold IP that computes the deletions from the sets $V(C^j)$, $j\in [m]$
    that minimize the number of pairs of connected vertices from the remaining graph, under the assumption that the guesses so far are correct.
    From all the solutions we will return the one that minimizes the number of pairs of connected vertices.
    Prior to presenting the $m$-fold IP, notice that since we have guessed that $S$ will be deleted,
    the number of deletions that we further do is at most $k-|S|$.
    We start by creating $\ell$ constants $n_i = |V_i|$, for $i \in [\ell]$.

    \proofsubparagraph{Local constraints.}
    Now we will consider each connected component $C^j$, $j \in [m]$ separately.
    Fix a $j \in [m]$.
    We say that the subset $X \subseteq V(C^j)$ is \emph{invalid} if
    there exist $s,t \in U \setminus S$ such that (i) $s \in V_{i_1}$, $t \in V_{i_2}$, and $i_1 \neq i_2$,
    and (ii) $s$ and $t$ belong to the same connected component of $G[(U \setminus S) \cup (V(C^j) \setminus X)]$.
    We say that a subset of $V(C^j)$ is \emph{valid} if it is not invalid.
    Notice that $\varnothing$ might be a valid subset of $V(C^j)$.

    For any given $C^j$, $j \in [m]$, let $S^j_1, \ldots, S^j_\mu$, where $\mu \le 2^{|V(C^j)|} \le 2^\vi$,
    be an enumeration of all valid subsets of $V(C^j)$.
    Notice that the intersection of an optimal solution with $V(C^j)$ is exactly one of the sets $S^j_1, \ldots, S^j_\mu$.
    In order to force this, we  will create a variable $I^j_q \in \{0,1\}$ for each $S^j_q$ that will indicate
    whether we select to delete the set $S^j_q$ or not.
    We will also add the following constraint that forces the IP to select at most one of these sets:
    \[
        \sum_{q\in [\mu]} I^j_q = 1.
    \]

    Next, for each $q \in [\mu]$ and $i \in [\ell]$, we compute the constant $c^{j,q}_i$ which is equal to the number of vertices of
    $V(C^j) \setminus S^j_q$ that are connected to vertices of $V_i$ in $G[V_i \cup (V(C^j)\setminus S^j_q)]$.
    This indicates that $c^{j,q}_i$ vertices of $V(C^j)$ will be connected to $V_i$ in the final graph if we delete $S^j_q$.

    We further compute for each $q \in [\mu]$ the constant $p^{j,q}$ which is equal to the number of pairs of connected vertices of $V(C^j)$
    that will appear in the final graph and which are not connected with any set $V_i$, $i \in [\ell]$.
    This can be done in time $\bO (\vi^2)$ as the graph $G[(U \setminus S) \cup (V(C^j) \setminus X)]$ has at most $\vi$ vertices.
    Notice that $p^{j,q} \le \binom{\vi}{2}$.

    Notice that the constraints $I^j_q$, $q \in [\mu]$, along with the constants $c^{j,q}_{i}$, $q \in [\mu]$, $i \in [\ell]$,
    allow us to compute the number of vertices of $C^j$ connected with each $V_i$ in the final graph, without knowing which set $S^j_q$ has been deleted.
    In particular, we know that the number of vertices that belong in $C^j$ and will be connected with $V_i$, for any $i \in [\ell]$, is exactly
    \[
        x^j_i = \sum_{q\in [\mu]} I^j_q c^{j,q}_i.
    \]

    Similarly, we can compute the number of pairs of connected vertices in $C^j$ after the deletions that are not connected with
    any set $V_i$, $i\in [\ell]$,
    which is
    \[
        p^j = \sum_{q \in [\mu]} I^j_q p^{j,q}.
    \]

    Finally, we introduce for each $q \in [\mu]$ a constant $d^j_q = |S^j_q|$ that accounts for the number of deletions
    made in $C^j$ when the intersection of the deletion set with $S^j$ is exactly $S^j_q$,
    that is, we set $d^j_q = |S^j_q|$.
    Consequently, regardless of the set we delete, the number of deletions in $C^j$ is exactly
    \[
        d^j = \sum_{q \in [\mu]} I^{j}_q d^j_q.
    \]

    We remark that, for each $j \in [m]$, we have $\bO (2^\vi)$ variables and $\bO (1)$ constraints.
    Also, all the coefficients are bounded by $\binom{\vi}{2}$.
    Furthermore, those can be considered as $j$ sets of variables and constraints as, for the moment,
    there is no constraint that includes variables related to distinct connected components $C^{j_1}$ and $C^{j_2}$, where $j_1 \neq j_2$.

    \proofsubparagraph{Global constraints.}
    In this part, we deal with the global variables and constraints.
    First we introduce the following constraint so that we force the number of deletions to be bounded:
    \[
        \sum_{j \in [m]} d^j = \sum_{j \in [m]} \sum_{q \in [\mu]} I^{j}_q d^j_q = k-|S|.
    \]
    Notice that $k-|S|$ is not a coefficient but a right-hand side constant, i.e., it is not multiplied by any variable. 

    Next, we create one variable $x_i$ for each $i \in [\ell]$ which represents the number of vertices that appear in the connected components
    each $V_i$ belongs to and we set
    \[
        x_i = n_i + \sum_{j\in [m]} x^j_i = n_i +\sum_{j\in [m]} \sum_{q\in [\mu]} I^j_q c^{j,q}_i.
    \]

    Lastly, we create a variable to account for the number of pairs of connected vertices which are not connected to any $V_i$, $i\in [\ell]$.
    Those are exactly:
    \[
        p = \sum_{j \in [m]} p^j = \sum_{j \in [m]} \sum_{q \in [\mu]} I^j_q p^{j,q}.
    \]

    Notice that the number of global constraints is $\bO(\vi)$ and the size of any coefficient is $\bO(\vi^2)$.

    It remains to define the objective function.
    In particular, we want to minimize the following:
    \[
        p + \sum_{i\in [\ell]} \binom{x_i}{2}.
    \]

    The above formulation is an $m$-fold IP with a separable convex objective function.
    Indeed, if we ignore the $\bO(\vi)$ global constraints,
    then we can partition the variables into $m$ (disjoint) sets of size $2^{\bO(\vi)}$ and
    the constraints into $m$ (disjoint) sets of size at most $\bO(1)$ such that,
    the variables of the $i$-th variable set appear only in the $i$-th set of constraints.
    Additionally, all the coefficients are at most $\vi^2$.
    Therefore, we can conclude that the $m$-fold IP runs in $\vi^{\bO(\vi^2)}n^{\bO(1)}$ time~\cite{KouteckyLO18}.
    Therefore, it remains to prove that this returns a solution of the {\CNC} problem.

    We argue that the IP computes a value that represent a feasible set of deletions.
    Notice that the IP always computes the number of connected pairs based on the assumption that
    the considered partition $V_1,\ldots,V_{\ell}$ of $U\setminus S$ satisfies the following two properties after the deletions:
    \begin{enumerate}
        \item for each $i \in [\ell]$, all vertices of $V_i$ belong to the same connected component,
        \item there is no connected component that intersects with more than one $V_i$, $i \in [\ell]$.
    \end{enumerate}
    Although we have guaranteed the second condition (by considering only valid subsets of each $C^j$), we have no guarantee about the first.
    In order to see that the algorithm will indeed return the minimum number of connected pairs,
    it suffices to observe the following.
    Given the sets $S^1_{q_i}, \ldots, S^m_{q_m}$ we will delete from each component $C^1, \ldots, C^m$,
    \begin{itemize}
        \item if both conditions hold in $G - \left( S \cup \bigcup_{j=1}^m S^j_{q_j} \right)$,
        then the optimization function computes the actual number of connected pairs,

        \item if only the second condition holds in $G - \left( S \cup \bigcup_{j=1}^m S^j_{q_j} \right)$,
        then the optimization function computes a number of connected pairs greater than or equal to the number of connected pairs.
    \end{itemize}
    Therefore, when we have made the correct guesses, the algorithm returns the optimal solution while,
    in all other cases, the algorithm returns a value greater than or equal to a feasible solution (which cannot be smaller than the optimal).
    Finally, we need to check whether the number of connected pairs we computed is at most $x$.

    \proofsubparagraph{Running time.}
    It remains to argue about the running time of our algorithm. We first compute the set $U$ in time $\vi^{\bO(\vi)}n^{\bO(1)}$~\cite{algorithmica/DrangeDH16}.
    Then we are guessing which vertices of $U$ will be included in an optimal deletion set,
    as well as the way the non-deleted vertices of $U$ are partitioned in connected components in the final graph.
    These guesses result in an $2^{\vi} \vi^\vi$ overhead. For each guess, we create an $m$-fold IP.
    
    Note that each $m$-fold IP requires $(a\cdot r \cdot s)^{(r^2\cdot s +r\cdot s^2)} (m\cdot t\cdot L)^{\bO(1)}$
    time~\cite[Corollary~97]{arxiv/EisenbrandHKKLO19}.%
    \footnote{This result is a direct consequence of the results published in~\cite{ipco/HunkenschroderKLV25}.}
    In this expression, $a$ denotes the largest coefficient,
    $r$ denotes the number of global constraints,
    while $s$ and $t$ denote the number of constraints and the number of variables respectively, in each block of local constraints.
    Finally, $L$ denotes the length of the binary encoding of the numerical input data of the $m$-fold IP.
    In our case we have that $a \le \vi^2$,
    $r \le \vi$,
    $s = \bO(1)$,
    and $t \le 2^\vi$.
    Furthermore, $m \le n$ and $L = \Theta \big( (r +  m \cdot s) \cdot m \cdot t \log{n} \big)$.
    Therefore, each $m$-fold IP runs in $\vi^{\bO(\vi^2)}n^{\bO(1)}$ time.
    In total, the algorithm runs in $\vi^{\bO(\vi^2)}n^{\bO(1)}$ time.
\end{proof}
\subsection{Modular-width}


All our previous results are concerned with \emph{sparse} classes of graphs.
Here we prove that {\CNC} is fixed-parameter tractable parameterized by the modular-width of the input
graph, thus improving over the analogous result for neighborhood diversity~\cite{mthesis/Schestag21}.

\begin{theoremrep}[\appsymb]\label{thm:CNC:mw}
    There is an algorithm that given any instance $\mathcal{I} = (G,k,x)$ of \CNC,
    it decides $\mathcal{I}$ in time $2^{\mw} n^{\bO(1)}$,
    where $\mw$ denotes the modular-width of $G$.
\end{theoremrep}

\begin{proof}
    We first recall some necessary definitions.
    A \emph{modular decomposition} of a graph $G$ is a rooted tree $T_G$ where any node $b$ with $c$ children is
    also associated with a graph $H_b$ consisting of $c$ vertices $v_1, \ldots, v_c$. 
    Additionally, any node $b$ of $T_G$ represents a graph $G_b = (V_b, E_b)$ defined recursively:
    \begin{itemize}
        \item a leaf node represents an isolated vertex,
        \item a non-leaf node $b$ with $c$ children $b_1, \ldots, b_c$ represents the graph $G_b = (V_b, E_b)$ where $V_b = \bigcup_{i\in [c]} V_{b_i}$ and $E_b = \setdef{u_i u_j}{u_i \in V_{b_i}, \, u_j \in V_{b_j} \text{ and } v_i v_j \in E(H_b)} \cup \bigcup_{i \in [c]} E_{b_i}$.
    \end{itemize}
    Finally, the graph $G_r$, where $r$ is the root of $T_G$, is isomorphic to $G$.

    The \emph{width} of a modular decomposition $T_G$ is equal to the maximum number of children of any node in $T_G$.
    The \emph{modular-width} of a graph $G$, denoted $\mw$, is the minimum width among the modular decompositions of $G$.
    Notice that, by definition, $T_G$ has exactly $n$ leaves.
    Additionally, we can assume that there is no node in $T_G$ with exactly one child;
    indeed, such a node would represent the same graph as its child node.
    Therefore, since $T_G$ is a tree, it has at most $2n-1$ nodes.

    Let $T_G$ be a modular decomposition of $G$ computed in polynomial time~\cite{dm/McConnellS99}.
    We will perform dynamic programming over the nodes of $T_G$.
    Let $b$ be a node of $T_G$.
    For the graph $G_b$, we will compute a function $f_b \colon [0,k] \to \mathbb{N}$ such that
    for $k' \in [0,|V_b|-1]$, $f_b(k')$ is equal to the minimum number of pairs of connected vertices
    in $G_b$ after $k'$ vertex deletions.
    Notice that for any graph $G$ there is no reason to delete more than $|V(G)|-1$ vertices,
    thus we set $f_b(k') = +\infty$ for all $k' \in [|V_b|, k]$.
    This will also be useful later when we will need to guarantee that we are not deleting the whole set $V_b$.
    We explain how to compute recursively the function $f_b$.

    If $b$ is a leaf node of $T_G$, then the graph $G_b$ is isomorphic to an isolated vertex,
    in which case $G_b$ has no pairs of connected vertices.
    Consequently, we set $f_b(0)=0$ and $f_b(k')=+\infty$ for all $k' \in [k]$.

    Alternatively, it holds that $b$ is a non-leaf node of $T_G$, and let $b_1, \ldots, b_{\ell}$, with $\ell \le \mw$,
    denote the nodes of $T_G$ that are children of $b$.
    Recall that $b$ is related to a graph $H_b = (V_b,E_b)$ with $V(H_b) = \{ v_1, \ldots, v_\ell\}$ and that in $G_b$,
    for all vertices $u_i \in V_{b_i}$ and $u_j \in V_{b_j}$ it holds that
    $u_i u_j \in E(G_b)$ if and only if $v_i v_j \in E(H_b)$.

    To deal with this case we adopt an exhaustive approach.
    Fix $S \subseteq [\ell]$ and let $U = \bigcup_{i \in S} V_{b_i}$.
    We will compute the function $f^S_b \colon [0, |V_b|-1] \to \mathbb{N}$ that, for $k' \in [0, |V_b|-1]$,
    returns the minimum number of pairs of connected vertices remaining in $G_b$ after $k'$ vertex deletions under the assumptions that:
    \begin{itemize}
        \item all vertices of $U$ are deleted, and
        \item no set $V_{b_i}$, $i \notin S$ is completely deleted.
    \end{itemize}
    Intuitively, set $S$ indicates whether a set of vertices $V_{b_i}$ is completely deleted or not.
    We set $f^S_b(k') = +\infty$ for all $k' < |U|$ as we cannot delete $U$ with less than $|U|$ deletions.

    We now consider values $k' \ge |U|$.
    Prior to explaining how to compute $f^S_b(k')$ we need the following definitions.
    Let $S' \subseteq [\ell] \setminus S$ such that,
    for each $i \in S'$ and $v \in V_{b_i}$ we have that $N_{G_b - U}(v)\subseteq V_{b_i}$.
    Notice that, by the definition of $S'$, any vertex $v \in V_{b_i}$, $i \in S'$, does not share any edge with any vertex in $V_b - (U \cup V_{b_i})$.
    Let $U' = \bigcup_{i \in S'} V_{b_i}$.
    We also need to consider the connected components of $G_b - (U \cup U')$.
    Let $C_1,\ldots,C_m$, $m \le \ell$, denote the vertex sets of the connected components of $G_b - (U \cup U')$.
    Due to the removal of $U$ and $U'$, we know that for any set $C \in \{ C_1, \ldots, C_m \}$
    there exist at least two integers $i,j \in [\ell] \setminus (S \cup S')$ such that $C \cap V_{b_i} \neq \varnothing$
    and $C \cap V_{b_j} \neq \varnothing$.

    Now, we are ready to explain how we compute the value $f^S_b(k')$, for $k' \ge |U|$.
    We know that $|U|$ of the deletions need to be used in order to delete the vertex set $U$.
    Assume that we know how the deletions are distributed between the modules $V_{b_i}$, $i \in S'$ and $C_j$, $j \in [m]$.
    In particular let $k_i$ be the number of deletions from each module $V_{b_i}$, $i \in S$ and
    $k'_j$ be the number of deletions from each connected component $C_j$, $j \in [m]$.
    Finally, notice that $k' = |U| + \sum_{i \in S'}k_i + \sum_{j\in [m]} k'_j$.

    In order to compute $f^S_b(k')$, we first deal with each graph $G[C_j]$, $j \in [m]$, separately.
    Let $f_j \colon [k] \to \mathbb{N}$ be a function such that, for any $k'_j\in [k]$,
    $f_j(k'_j)$ equals the minimum number of pairs of connected vertices remaining in $G[C_j]$
    after $k'_j$ vertex deletions under the assumption that,
    for any child $b_i$, $i \in [\ell]$ of $b$ such that $V_{b_i} \cap C_j \neq \varnothing$, $V_{b_i}$ is not completely deleted.
    Notice that this assumption comes from the fact that we have previously guessed which modules will be completely deleted
    (set $S$, with the union of the vertex sets of the deleted modules denoted by $U$).
    Notice that, $C_j$ is a connected component of $G_b-(U\cup U')$ so it does not include vertices from any module which will be completely deleted.

    Since for any $j \in [m]$ we have assumed that there is always at least one vertex of each module that remains in $C_j$ after all deletions
    and $G[C_j]$ is connected,
    we have that after the deletions this component remains connected.
    Let $\mu_j$ be the number of modules that intersect $C_j$.
    If $k'_j > |C_j| - \mu_j$ it follows that we cannot respect the assumption that any module that intersects $C_j$ will not be completely deleted,
    therefore we set $f_j(k'_j) = +\infty$ for all $k'_j > |C_j| - \mu_j$.
    For $k'_j \in [0, |C_j| - \mu_j]$, we set $f_j(k'_j) = \binom{|C_j|-k'_j}{2}$;
    indeed, by our assumption, $C_j$ must remain connected after the deletions if $k'_j \in [0, |C_j| - \mu_j]$,
    thus $\binom{|C_j|-k'_j}{2}$ is the exact number of connected pairs remaining in $C_j$ after $k'_j$ deletions.
    Furthermore, if $f_j(k'_j) \neq +\infty$, there exists a deletion set $D \subseteq C_j$ such that $|D| = k'_j$ and
    $G[C_j \setminus D]$ has exactly $f_j(k'_j)$ pairs of connected vertices.

    Now, let us consider pairs of connected vertices in $G_b - U$ that include vertices of $V_{b_i}$, for some $i \in S'$.
    Notice that by the definition of $S'$, both vertices of such pairs belong to $V_{b_i}$.
    Since there are $k_i$ deletions taking place in $V_{b_i}$,
    we can conclude that the minimum number of such pairs of connected vertices after $k_i$ deletions is $f_{b_i}(k_i)$.

    Therefore, we can conclude that
    \[
        f^S_b(k) =
        \min_{k_i \ge 0, i \in S' \text{ and } k'_j \ge 0, j \in [m]}
            \setdef*{ \sum_{i \in S'} f_{b_i}(k_i) +\sum_{j\in [m]} f_j(k'_j)}{k = |U| + \sum_{i \in S'} k_i + \sum_{j\in [m]} k'_j}.
    \]
    Notice that, as far as the value is not infinite, we can guarantee that there exists a solution that agrees with the restrictions implied by $S$.

    To compute the function $f_b$ assume that we have computed the functions $f^S_b$, for all $S \subseteq [\ell]$.
    Then we simply set $f_b (k') = \min_{S\subseteq [\ell]} f^S_b(k')$.
    Since we made all guesses $S \subseteq [\ell]$, we know that we have included the guess that agrees with the optimal solution.
    Let $S^*$ be this guess. Notice that $f^{S^*}_b(k')$ is equal to the minimum number of pairs of connected vertices in $G[V_b]$ after $k'$ deletions.
    Also, for any other $S'$, if $f^{S'}_b(k') \neq +\infty$, we have that $f^{S'}_b(k')$ equals the number of pairs that we can achieve if the deletions agree with $S'$.
    Since this is indeed a solution we have that $f^{S'}_b(k') \ge f^{S^*}_b(k')$. Therefore, we always return the minimum number of pairs.

    \proofsubparagraph{Running time.}
    We now argue about the running time of our algorithm.
    Let $b$ denote a non-leaf node of $T_G$.
    We argue that we can compute $f_b$ in time $2^\mw n^{\bO(1)}$.
    Notice that after fixing the set $S$, we can compute $S'$, the sets $C_1,\ldots C_m$, and the functions $f_i$, for all $i \in [m]$, in polynomial time.
    Also, all $f_{b_i}$, $i \in S$, have been recursively computed in the children nodes of $b$.
    Finally, by \Cref{lemma:minimize_functions}, we know that given the functions $f_{b_i}$, $i \in S'$ and $f_j$, $j \in [m]$,
    we can compute $f^S_b$ in polynomial time.
    Notice that there are at most $2^{\mw}$ functions  $f^S_b$, since $S \subseteq [\ell] \subseteq [\mw]$.
    Since each function is computed in polynomial time, in total we need $2^{\mw} n^{\bO(1)}$ time to compute them all.
    Furthermore, since $f_b (k') = \min_{S \subseteq [\ell]} f^S_b(k')$, we additionally need $2^{\mw} n$ time.
    Finally, observe that, after we compute the function $f_r$ of the root node $r$, we can compute the minimum number of connected pairs after $k$ deletions by checking the value $f_r(k)$.
    Since for every node of the modular decomposition we need $2^{\mw}n^{\bO(1)}$ time in the worst case,
    and the modular decomposition has at most $2n-1$ nodes, the total running time is $2^{\mw}n^{\bO(1)}$.
\end{proof}

\subsection{Clique-width}

\Cref{thm:CNC:fes} implies that {\CNC} is W[1]-hard parameterized by clique-width.
Here we show that for this parameterization, the problem belongs to XP.
To obtain this result, we proceed by presenting a standard dynamic programming algorithm.
On a high-level, as we go over the clique-width expression of the input graph,
we consider a partition of the connected components of our graph based on the labels
appearing in their vertices. For every such partitioning, we keep track
of \emph{both} the total size as well as the total number of connected pairs of vertices
for the connected components belonging to the partitioning. As we prove, this amount of
information is sufficient in order to solve \CNC.

\begin{theorem}\label{thm:cw}
    There is an algorithm that,
    given a graph $G$ along with an irredundant clique-width expression $\psi$ of $G$ of width $\cw$,
    determines for all $k \in [0,n]$ the number of subsets $S \subseteq V(G)$ with $|S| = k$ that minimize
    the number of connected pairs of vertices of $G - S$,
    in time $n^{\bO(2^\cw)}$.
\end{theorem}

\begin{proof}
    Before describing our algorithm we start with some definitions and notations.
    Let $\mathcal{H}$ denote the set of all $\cw$-labeled graphs generated by subexpressions of $\psi$.
    Furthermore, let $\mathcal{L} = 2^{[\cw]} \setminus \varnothing$ denote the set of all non-empty subsets of $[\cw]$.
    We say that $\sigma = (\alpha, \beta)$ is a \emph{signature},
    where $\sigma$ is a tuple consisting of functions $\alpha, \beta \colon \mathcal{L} \to \mathbb{N}$.
    Moreover, let $\Sigma$ denote the set of all possible signatures, where $|\Sigma| = n^{\bO(2^\cw)}$.
    We say that the \emph{label set} of a connected component of a labeled graph is the set of labels appearing in the component's vertices.
    For $H \in \mathcal{H}$ and $L \in \mathcal{L}$, let $\mathcal{C}^L_H \subseteq \cc(H)$ denote the set of connected components of $H$ with label set $L$.
    Lastly, for $H \in \mathcal{H}$, we define the function $\sgn_H \colon 2^{V(H)} \to \Sigma$,
    such that for $S \subseteq V(H)$, the \emph{$H$-signature of $S$} is the tuple $\sgn_H(S) = \sigma = (\alpha, \beta) \in \Sigma$,
    where for all $L \in \mathcal{L}$ it holds that
    \begin{itemize}
        \item $\alpha(L)$ is equal to the sum of the sizes of the connected components of $\mathcal{C}^L_{H-S}$,

        \item $\beta(L)$ is equal to the number of pairs of connected vertices belonging to connected components of $\mathcal{C}^L_{H-S}$.
    \end{itemize}

    Our algorithm proceeds by dynamic programming along $\psi$ in a bottom-up fashion.
    In particular, for every $H \in \mathcal{H}$ it stores a table $\DP_H[\cdot,\cdot]$ where,
    as we show, for $k \in [0,n]$ and $\sigma \in \Sigma$,
    $\DP_H[k,\sigma]$ is equal to the number of subsets of $V(H)$ of size $k$
    and $H$-signature $\sigma$.
    We now proceed to describe how to populate the DP tables,
    as well as establish this invariant by induction.

    \proofsubparagraph{Singleton $H = i(v)$.}
    For $H = i(v)$, notice that $\lab^{-1}_H(i) = \{v\}$ and $\lab^{-1}_H(w) = \varnothing$ for all $w \in [\cw] \setminus \{i\}$.
    Consequently, for $k \in [0,n]$ and signature $\sigma = (\alpha, \beta) \in \Sigma$ we set
    \begin{equation}\label{eq:cw:singleton}
        \DP_H[k,\sigma] =
        \begin{cases}
            1   &\text{if $k=0$, $\alpha \equiv \alpha_0[\{i\} \mapsto 1]$, and $\beta \equiv \beta_0$},\\
            1   &\text{if $k=1$, $\alpha \equiv \alpha_0$, and $\beta \equiv \beta_0$},\\
            0   &\text{otherwise},
        \end{cases}
    \end{equation}
    where for all $L \in \mathcal{L}$,  $\alpha_0(L) = \beta_0(L) = 0$.
    By \Cref{eq:cw:singleton} we can fill the table $\DP_H[\cdot,\cdot]$ in time $n^{\bO(2^\cw)}$.
    Furthermore, notice that $2^{V(H)} = \{ \varnothing, \{v\} \}$ and
    one can easily verify that the invariant holds.

    \proofsubparagraph{Disjoint union, $H = H_1 \oplus H_2$.}
    We define a function $h \colon \Sigma \times \Sigma \to \Sigma$
    such that, for $\sigma_1 = (\alpha_1, \beta_1)$ and $\sigma_2 = (\alpha_2, \beta_2)$,
    $h(\sigma_1,\sigma_2) = \sigma = (\alpha,\beta)$
    where for all $L \in \mathcal{L}$ it holds that
    \begin{itemize}
        \item $\alpha(L) = \alpha_1(L) + \alpha_2(L)$,
        \item $\beta(L) = \beta_1(L) + \beta_2(L)$.
    \end{itemize}
    For a fixed size $k \in [0,n]$ and signature $\sigma \in \Sigma$,
    we set the value of $\DP_H[k,\sigma]$ to be
    \begin{equation}\label{eq:cw:union}
        \DP_H[k,\sigma] = \sum_{k_1 + k_2 = k} \, \sum_{(\sigma_1, \sigma_2) \in h^{-1}(\sigma)} \DP_{H_1}[k_1,\sigma_1] \cdot \DP_{H_2}[k_2,\sigma_2].
    \end{equation}
    If $h^{-1}(\sigma) = \varnothing$ we set $\DP_H[k,\sigma] = 0$.
    Notice that computing $h^{-1}(\sigma)$ requires $n^{\bO(2^\cw)}$ time,
    thus by \cref{eq:cw:union} we can fill the table $\DP_H[\cdot,\cdot]$ in time $n^{\bO(2^\cw)}$.

    \begin{claim}\label{claim:cw:union}
        Let $H = H_1 \oplus H_2$.
        Assume that for $H_i \in \{ H_1, H_2 \}$, $k \in [0,n]$, and $\sigma \in \Sigma$,
        $\DP_{H_i}[k,\sigma]$ is equal to the number of subsets of
        $V(H_i)$ of size $k$ and $H_i$-signature $\sigma$.
        Then it holds that $\DP_H[k,\sigma]$ is equal to the number of subsets of
        $V(H)$ of size $k$ and $H$-signature $\sigma$.
    \end{claim}
    
    \begin{claimproof}
        Observe that $V(H) = V(H_1) \cup V(H_2)$ and $E(H) = E(H_1) \cup E(H_2)$.
        Thus, for every $S \subseteq V(H)$, it holds that $S = S_1 \cup S_2$ where $S_i = S \cap V(H_i)$ for $i \in \{1,2\}$.
        Furthermore, notice that the connected components of $H - S$ are exactly the connected components of $H_1 - S_1$ and $H_2 - S_2$,
        each with the same label set.
        Consequently, for every $L \in \mathcal{L}$ it holds that
        \begin{itemize}
            \item $\alpha(L) = \alpha_1(L) + \alpha_2(L)$,
            \item $\beta(L) = \beta_1(L) + \beta_2(L)$,
        \end{itemize}
        where $\sigma = (\alpha,\beta)$ is the $H$-signature of $S$,
        and $\sigma_i = (\alpha_i,\beta_i)$ is the $H_i$-signature of $S_i$ for $i \in \{1,2\}$.
        This implies that $\sigma = h(\sigma_1,\sigma_2)$.
        The claim follows by \cref{eq:cw:union} and by the fact that there are $\DP_{H_1}[k_1,\sigma_1]$ ways to choose a subset of size $k_1$
        with $H_1$-signature $\sigma_1$, and $\DP_{H_2}[k_2,\sigma_2]$ ways to choose a subset of size $k_2$
        with $H_2$-signature $\sigma_2$, for every pair $(\sigma_1,\sigma_2) \in h^{-1}(\sigma)$ and every pair $(k_1,k_2)$ such that $k = k_1 + k_2$.
    \end{claimproof}

    \proofsubparagraph{Relabeling, $H = \rho_{i \to j}(H')$.}
    Observe that a connected component with label set $L \in \mathcal{L}$ in $H$ has label set either $L$, $L \cup \{i\}$, or $L \cup \{i\} \setminus \{j\}$ in $H'$.
    For each pair of distinct $i,j \in [\cw]$
    we define a function $f_{i \to j} \colon \Sigma \to \Sigma$
    such that, for $\sigma' = (\alpha', \beta')$, $f_{i \to j}(\sigma') = \sigma = (\alpha,\beta)$
    where for all $L \in \mathcal{L}$ it holds that
    \begin{itemize}
        \item if $i \in L$, then $\alpha(L) = \beta(L) = 0$,
        \item if $i,j \notin L$, then $\alpha(L) = \alpha'(L)$ and $\beta(L) = \beta'(L)$,
        \item if $i \notin L$ and $j \in L$, then it holds that
        \begin{itemize}
            \item $\alpha(L) = \alpha'(L) + \alpha'(L \cup \{i\}) + \alpha'(L \cup \{i\} \setminus \{j\})$,
            \item $\beta(L) = \beta'(L) + \beta'(L \cup \{i\}) + \beta'(L \cup \{i\} \setminus \{j\})$.
        \end{itemize}
    \end{itemize}
    For a fixed size $k \in [0,n]$ and signature $\sigma \in \Sigma$,
    we set the value of $\DP_H[k,\sigma]$ to be
    \begin{equation}\label{eq:cw:relabel}
        \DP_H[k,\sigma] = \sum_{\sigma' \in f^{-1}_{i \to j}(\sigma)} \DP_{H'}[k,\sigma'].
    \end{equation}
    If $f^{-1}_{i \to j}(\sigma) = \varnothing$ we set $\DP_H[k,\sigma] = 0$.
    Notice that computing $f^{-1}_{i \to j}(\sigma)$ requires $n^{\bO(2^\cw)}$ time,
    thus by \cref{eq:cw:relabel} we can fill the table $\DP_H[\cdot,\cdot]$ in time $n^{\bO(2^\cw)}$.

    \begin{claim}\label{claim:cw:relabel}
        Let $H = \rho_{i \to j}(H')$.
        Assume that for $k \in [0,n]$ and $\sigma \in \Sigma$,
        $\DP_{H'}[k,\sigma]$ is equal to the number of subsets of
        $V(H')$ of size $k$ and $H'$-signature $\sigma$.
        Then it holds that $\DP_H[k,\sigma]$ is equal to the number of subsets of
        $V(H)$ of size $k$ and $H$-signature $\sigma$.
    \end{claim}

    \begin{claimproof}
        Observe that $V(H) = V(H')$ and $E(H) = E(H')$.
        Furthermore, notice that for all $S \subseteq V(H)$, the connected components of $H - S$ are exactly the connected components of $H' - S$,
        albeit with potentially different label sets.
        Let $S \subseteq V(H)$ such that $\sgn_H(S) = \sigma = (\alpha,\beta)$ and
        $\sgn_{H'}(S) = \sigma' = (\alpha',\beta')$.
        We argue that $f_{i \to j}(\sigma') = \sigma$.
        Given that each subset $S$ has a unique $H'$-signature,
        \cref{eq:cw:relabel} and the fact that there are $\DP_{H'}[k,\sigma']$ ways to choose a subset of size $k$
        with $H'$-signature $\sigma'$,
        this implies the claim.

        Let $L \in \mathcal{L}$.
        If $i \in L$, then since $\lab^{-1}_H(i) = \varnothing$,
        it holds that no connected component of $H - S$ has label set $L$,
        thus $\alpha(L) = \beta(L) = 0$.
        If $i,j \notin L$, then the connected components of $H - S$ with label set $L$ are exactly the connected components of $H' - S$ with label set $L$,
        thus $\alpha(L) = \alpha'(L)$ and $\beta(L) = \beta'(L)$.
        Lastly, if $i \notin L$ and $j \in L$, then
        the connected components of $H - S$ with label set $L$ are exactly the
        connected components of $H' - S$ with label set $L$, $L \cup \{i\}$, and $L \cup \{i\} \setminus \{j\}$.
        Thus, $\alpha(L) = \alpha'(L) + \alpha'(L \cup \{i\}) + \alpha'(L \cup \{i\} \setminus \{j\})$ and
        $\beta(L) = \beta'(L) + \beta'(L \cup \{i\}) + \beta'(L \cup \{j\} \setminus \{j\})$.
        It follows that $f_{i \to j}(\sigma') = \sigma$,
        and this concludes the proof.
    \end{claimproof}

    \proofsubparagraph{Joining labels with edges, $H = \eta_{i,j}(H')$.}
    For each pair of distinct $i,j \in [\cw]$
    we define a function $g_{i,j} \colon \Sigma \to \Sigma$ as follows.
    Let for $\sigma' = (\alpha', \beta')$, $g_{i,j}(\sigma') = \sigma = (\alpha,\beta)$.
    We consider two cases.
    First, assume that either $\sum_{L \in \mathcal{L}_i} \alpha'(L) = 0$ or $\sum_{L \in \mathcal{L}_j} \alpha'(L) = 0$,
    where for $z \in \{i,j\}$, $\mathcal{L}_z = \setdef{L \in \mathcal{L}}{z \in L}$ denotes the label sets containing label $z$.
    In that case, we set $\sigma = \sigma'$.
    Otherwise, let $\mathcal{L}^{\sigma'}_{i,j} = \setdef{L \in \mathcal{L}}{\{i,j\} \cap L \neq \varnothing \text{ and } \alpha'(L) \neq 0}$
    and $L^{\sigma'}_{i,j} = \bigcup_{L \in \mathcal{L}^{\sigma'}_{i,j}} L$.
    For all $L \in \mathcal{L}$ we set the values of $\alpha(L)$ and $\beta(L)$ as follows:
    \begin{itemize}
        \item if $i,j \notin L$, then $\alpha(L) = \alpha'(L)$ and $\beta(L) = \beta'(L)$,
        \item if $\{ i,j \} \cap L \neq \varnothing$ and $L \neq L^{\sigma'}_{i,j}$, then $\alpha(L) = \beta(L) = 0$,
        \item if $L = L^{\sigma'}_{i,j}$, then $\alpha(L) = \sum_{L \in \mathcal{L}^{\sigma'}_{i,j}} \alpha'(L)$ and
            $\beta(L) = \binom{\alpha(L)}{2}$.
    \end{itemize}
    For a fixed size $k \in [0,n]$ and signature $\sigma \in \Sigma$,
    we set the value of $\DP_H[k,\sigma]$ to be
    \begin{equation}\label{eq:cw:join}
        \DP_H[k,\sigma] = \sum_{\sigma' \in g^{-1}_{i,j}(\sigma)} \DP_{H'}[k,\sigma'].
    \end{equation}
    If $g^{-1}_{i,j}(\sigma) = \varnothing$ we set $\DP_H[k,\sigma] = 0$.
    Notice that computing $g^{-1}_{i,j}(\sigma)$ requires $n^{\bO(2^\cw)}$ time,
    thus by \cref{eq:cw:join} we can fill the table $\DP_H[\cdot,\cdot]$ in time $n^{\bO(2^\cw)}$.

    \begin{claim}\label{claim:cw:join}
        Let $H = \eta_{i,j}(H')$.
        Assume that for $k \in [0,n]$ and $\sigma \in \Sigma$,
        $\DP_{H'}[k,\sigma]$ is equal to the number of subsets of
        $V(H')$ of size $k$ and $H'$-signature $\sigma$.
        Then it holds that $\DP_H[k,\sigma]$ is equal to the number of subsets of
        $V(H)$ of size $k$ and $H$-signature $\sigma$.
    \end{claim}

    \begin{claimproof}
        Observe that $V(H) = V(H')$, each vertex has the same label in both $H$ and $H'$,
        and $E(H) = E(H') \cup \setdef{uv}{u \in \lab^{-1}_{H'}(i), v \in \lab^{-1}_{H'}(j)}$.
        Let $S \subseteq V(H)$ such that $\sgn_H(S) = \sigma = (\alpha,\beta)$ and
        $\sgn_{H'}(S) = \sigma' = (\alpha',\beta')$.
        We argue that $g_{i,j}(\sigma') = \sigma$.
        Given that each subset $S$ has a unique $H'$-signature,
        \cref{eq:cw:join} and the fact that there are $\DP_{H'}[k,\sigma']$ ways to choose a subset of size $k$
        with $H'$-signature $\sigma'$,
        this implies the claim.
        Let $L \in \mathcal{L}$.
        We will consider multiple cases.

        First, assume that either $\sum_{L \in \mathcal{L}_i} \alpha'(L) = 0$ or $\sum_{L \in \mathcal{L}_j} \alpha'(L) = 0$,
        where for $z \in \{i,j\}$, $\mathcal{L}_z = \setdef{L \in \mathcal{L}}{z \in L}$ denotes the label sets containing label $z$.
        In that case, it follows that either $\lab^{-1}_{H'}(i) \subseteq S$ or $\lab^{-1}_{H'}(j) \subseteq S$.
        Consequently, no new edges are added in $H - S$ compared to $H' - S$,
        thus, for all $L \in \mathcal{L}$, it holds that $\alpha(L) = \alpha'(L)$ and $\beta(L) = \beta'(L)$,
        which implies that $\sigma = \sigma'$.

        Alternatively it holds that $\sum_{L \in \mathcal{L}_i} \alpha'(L) \neq 0$ and $\sum_{L \in \mathcal{L}_j} \alpha'(L) \neq 0$.
        If $i,j \notin L$, then the connected components of $H - S$ with label set $L$ are exactly the connected components of $H' - S$ with label set $L$,
        thus $\alpha(L) = \alpha'(L)$ and $\beta(L) = \beta'(L)$ follows.
        To handle the remaining cases, we let $\mathcal{L}^{\sigma'}_{i,j} = \setdef{L \in \mathcal{L}}{\{i,j\} \cap L \neq \varnothing \text{ and } \alpha'(L) \neq 0}$
        and $L^{\sigma'}_{i,j} = \bigcup_{L \in \mathcal{L}^{\sigma'}_{i,j}} L$.
        Since there are vertices $u$ and $v$ in $H' - S$ with labels $i$ and $j$, respectively,
        it follows that in $H-S$, all vertices of label $i$ and all vertices of label $j$ belong to the same
        connected component.
        Furthermore, the label set of that connected component is exactly $L^{\sigma'}_{i,j}$,
        as it contains all labels of connected components of $H' - S$ with label set in $\mathcal{L}^{\sigma'}_{i,j}$.
        As for its size, it holds that $\alpha(L^{\sigma'}_{i,j}) = \sum_{L \in \mathcal{L}^{\sigma'}_{i,j}} \alpha'(L)$,
        from which we can infer the number of pairs of connected vertices $\beta(L^{\sigma'}_{i,j}) = \binom{\alpha(L^{\sigma'}_{i,j})}{2}$.
        For any other label set $L \in \mathcal{L}$ with $\{i,j\} \cap L \neq \varnothing$ and $L \neq L^{\sigma'}_{i,j}$,
        it holds that $\alpha(L) = \beta(L) = 0$ as either (i) $\alpha'(L) = \beta'(L) = 0$,
        or (ii) the connected components of $\mathcal{C}^L_{H'-S}$ are all connected in $H-S$ and are part of the single connected component
        of label set $L^{\sigma'}_{i,j}$ in $H-S$.
        It follows that $g_{i,j}(\sigma') = \sigma$,
        and this concludes the proof.
    \end{claimproof}
    Correctness follows by induction and \Cref{claim:cw:union,claim:cw:relabel,claim:cw:join}.
    As for the running time, notice that for every $H \in \mathcal{H}$
    the table $\DP_H[\cdot,\cdot]$ is filled in time $n^{\bO(2^\cw)}$.
    Since $|\mathcal{H}| = n^{\bO(1)}$,
    the overall running time of our algorithm is $n^{\bO(2^\cw)}$.
    Finally, notice that for any $S \subseteq V(G)$ of size $k \in [0,n]$ and $G$-signature $\sigma = (\alpha,\beta)$,
    it holds that $\pairs(G-S) = \sum_{L \in \mathcal{L}} \beta(L)$.
    Consequently, by iterating over all signatures $\sigma \in \Sigma$,
    we can determine the minimum value of $\pairs(G-S)$ over all subsets $S \subseteq V(G)$ of size $k$,
    as well as the number of such subsets achieving this minimum.
    This concludes the proof.
\end{proof}

\section{FPT Approximation Scheme}\label{sec:fpt-as}
Given the fact that, as evidenced by \Cref{thm:CNC:fes}, {\CNC} remains W[1]-hard even under severe structural parameterizations,
in this section we aim to bypass this computational hardness by adding approximation into the mix.
In particular, we design an efficient FPT-AS for the parameterization by treewidth by modifying the standard
$n^{\bO(\tw)}$ DP algorithm~\cite{dam/AddisSG13}
and making use of a technique introduced by Lampis~\cite{icalp/Lampis14}.

\begin{theoremrep}[\appsymb]\label{thm:fpt-as}
    There is an algorithm which, for all $\varepsilon > 0$,
    when given as input a graph $G$ of treewidth $\tw$
    returns in time $(\tw / \varepsilon)^{\bO(\tw)} n^{\bO(1)}$
    a set $S \subseteq V(G)$ of size at most $k$ such that
    $\pairs(G-S) \le (1+\varepsilon) \cdot \pairs(G-S')$ for all $S' \subseteq V(G)$ of size $|S'| \le k$,
    for all $k \le n$.
\end{theoremrep}

\begin{proof}
    We first describe the main idea behind our algorithm.
    On a high level, we aim to develop a DP which, while traversing the tree decomposition,
    keeps track of the sizes of any \emph{active} components (those whose vertices intersect the bag),
    while for the rest of the components (i.e., the \emph{inactive} ones)
    there exists a variable on which we account for their number of pairs of connected vertices.
    To this end, assuming that the exact size of an active component is $c$, our DP stores a value $\hat{c} \le (1+\varepsilon') \cdot c$
    where $\varepsilon'$ is such that $\binom{\hat{c}}{2} \le (1+\varepsilon) \cdot \binom{c}{2}$.
    Notice that for this to hold, $c=1$ implies that $\hat{c}=1$.

    \begin{lemma}\label{lemma:fpt-as:varepsilon}
        Let $0 < \varepsilon < 1$ and $c \ge 2$.
        For $\varepsilon' = \varepsilon/4$ it holds that
        if $\hat{c} \le (1 + \varepsilon') \cdot c$,
        then $\binom{\hat{c}}{2} \le (1+\varepsilon) \cdot \binom{c}{2}$.
    \end{lemma}

    \begin{nestedproof}
        It suffices to show that
        \begin{align*}
            \frac{(1 + \varepsilon') \cdot c \cdot ((1 + \varepsilon') \cdot c - 1)}{2} &\le
            \frac{(1 + \varepsilon) \cdot c \cdot (c - 1)}{2} \iff\\
            (1 + \varepsilon')^2 \cdot c - (1 + \varepsilon') &\le (1 + \varepsilon) \cdot c - (1 + \varepsilon) \iff\\
            c \cdot (\varepsilon - 2 \varepsilon' - (\varepsilon')^2) &\ge \varepsilon - \varepsilon'.
        \end{align*}

        Let $\varepsilon' = \varepsilon / x$ for some $x > 1$.
        In that case multiplying everything with $x^2 > 0$ in the previous inequality gives
        \begin{align*}
            c \cdot (x^2 \cdot \varepsilon - 2x \cdot \varepsilon - \varepsilon^2) &\ge x^2 \cdot \varepsilon - x \cdot \varepsilon \iff\\
            c \cdot (x^2 - 2x - \varepsilon) &\ge x^2 - x.
        \end{align*}

        Since $\varepsilon < 1$, it holds that $x^2 - 2x - \varepsilon > x^2 - 2x - 1$,
        thus, since $c \ge 0$, it suffices to show that $c \cdot (x^2 - 2x - 1) \ge x^2 - x$
        where $x^2 - 2x - 1 > 0$ for all $x \ge 3$.
        In that case, assuming that $x \ge 3$, it suffices to show that
        \[
            c \ge \frac{x^2 - x}{x^2 - 2x - 1}.
        \]
        Since $c \ge 2$,
        it suffices to have that $\frac{x^2 - x}{x^2 - 2x - 1} \le 2$,
        which is true for all $x \ge 4$.
    \end{nestedproof}

    Consequently, it suffices to present a dynamic program which correctly stores the size of any singleton active component,
    while for the rest of active components it allows for $(1 + \varepsilon')$-approximate values on their sizes.
    In that case, the number of pairs of connected vertices accounted for every connected component is a $(1+\varepsilon)$-approximation,
    and since we sum over those, the final value has a $(1+\varepsilon)$-approximation ratio as well.
    In the rest of the proof we present a DP that does exactly as required.

    \proofsubparagraph{Definitions and Notation.}
    We assume familiarity with the definition and usual notation for treewidth (see e.g.~\cite[Section~7]{books/CyganFKLMPPS15}).
    Let $\mathcal{T}$ denote the nice tree decomposition and $X_t$ denote the bag of node $t$, while $h_t$ denotes its height.
    We denote by $\mathcal{T}_t$ the subtree of $\mathcal{T}$ rooted at $t$,
    and by $G[\mathcal{T}_t]$ the subgraph due to the subtree $\mathcal{T}_t$ of the tree decomposition.
    For a node $t$, we denote by $X^{\downarrow}_t$ the union of the bags of the nodes of $\mathcal{T}_t$.
    %
    %
    A \emph{partition} $P$ of a set $L$ is a collection of non-empty subsets of $L$
    that are pairwise non-intersecting and such that $\bigcup_{p \in P} p = L$;
    each set in $P$ is called a \emph{block} of $P$.
    The set of partitions of a finite set $L$ is denoted by $\Pi(L)$,
    and $(\Pi(L),\sqsubseteq)$ forms a lattice where $P \sqsubseteq Q$
    if for each block $p$ of $P$ there is a block $q$ of $Q$ with $p \subseteq q$.
    The join operation of this lattice is denoted by $\sqcup$.
    For example, we have $\{ \{ 1,2\},\{3,4\},\{5\}\} \sqcup \{\{1\},\{2,3\},\{4\},\{5\}\}= \{ \{ 1,2,3,4\}, \{5\}\}$.
    Let $|P|$ denote the number of blocks of a partition $P$.
    For $P \in \Pi(L)$ and $X \subseteq L$, let $P_{\downarrow X} \in \Pi(X)$ be the partition
    $\setdef{p \cap X}{p \in P} \setminus \{\varnothing\}$,
    and for a set $Y$, let $P_{\uparrow Y} \in \Pi(L \cup Y)$ be the partition
    $P \cup \left(\bigcup_{y \in Y \setminus L} \{\{y\}\}\right)$.

    We are now ready to describe our algorithm.
    First we compute a nice tree decomposition of $G$ of width at most $2 \tw + 1$
    and at most $\bO(n \tw)$ bags, which can be computed in time $2^{\bO(\tw)} \cdot n$~\cite{focs/Korhonen21}.
    Next we apply a result of Bodlaender and Hagerup~\cite[Lemma~2.2]{siamcomp/BodlaenderH98} to said nice tree decomposition,
    which allows us to edit it into a tree decomposition of height $\bO(\log n)$, while the width remains $\bO(\tw)$;
    we further edit it so that it is \emph{nice}, in which case the height becomes $h = \bO(\tw \log n)$.
    Let $\delta  = \varepsilon' / (2h)$ and define $B = \{ 0 \} \cup \setdef{(1+\delta)^j}{j \in \mathbb{N}_0 \text{ and } (1+\delta)^j \le (1+\varepsilon') \cdot n}$.
    Informally, $B$ is the set of rounded values that are used to approximate the sizes of the active components.

    Our algorithm proceeds by dynamic programming along the nodes of $\mathcal{T}$ in a bottom-up fashion.
    For every node $t$ we define the \emph{pseudo-signature space} $\hat{\Sigma}_t$ which contains tuples $(Z, k, P, \hat{c})$
    where $Z \subseteq X_t$, $k \in [0,n]$,
    $P \in \Pi(Z)$ defines a partition of the vertices of $Z$,
    and $\hat{c} \colon P \to B$ is a function from the blocks of $P$ to $B$.
    We say that a set $S \subseteq X^{\downarrow}_t$ has \emph{pseudo-signature} (in node $t$) $\hat{\sigma}_t = (Z, k, P, \hat{c}) \in \hat{\Sigma}_t$
    if and only if
    \begin{romanenumerate}
        \item $Z = X_t \setminus S$,

        \item $|S \setminus X_t| = k$,

        \item $P$ describes the partition of the vertices of $Z$ into connected components of $G[\mathcal{T}_t] - S$,

        \item for all $p \in P$, $\hat{c}(p)$ is the minimum element of $B$ that upper-bounds the size of
        the component of $G[\mathcal{T}_t] - S$ that contains the vertices of $p$, without accounting for the vertices in $X_t$.
    \end{romanenumerate}
    In an analogous way, we define the \emph{signature space} $\Sigma_t$ and the
    \emph{signature} $\sigma_t = (Z,k,P,c)$ of a set $S \subseteq X^{\downarrow}_t$ (in node $t$),
    with the only difference being that $c \colon P \to [0,n]$ stores the \emph{exact} size of the active components,
    without accounting for the vertices in $X_t$.
    We further define for all $t \in V(\mathcal{T})$ the functions $\sgn_t \colon 2^{X^{\downarrow}_t} \to \Sigma_t$
    and $\hat{\sgn}_t \colon 2^{X^{\downarrow}_t} \to \hat{\Sigma}_t$ that map a set $S \subseteq X^{\downarrow}_t$ to
    its signature and pseudo-signature in node $t$ respectively.
    For a node $t$, we say that $\sigma_t$ and $\hat{\sigma}_t$ are \emph{compatible} when they only differ in the last entry of the tuple,
    that is, the only difference is between $c$ and $\hat{c}$.
    We say that a pair $(\sigma_t, d)$ is an \emph{exact solution} if there exists a set $S \subseteq X^{\downarrow}_t$,
    which we say that \emph{realizes} the exact solution,
    such that $\sgn_t(S) = \sigma_t$,
    and the number of pairs of connected vertices in the \emph{inactive} components of $G[\mathcal{T}_t] - S$ is equal to $d$.
    Analogously, we say that a pair $(\hat{\sigma}_t, \hat{d})$ is an \emph{approximate solution} if there exists a set $S \subseteq X^{\downarrow}_t$,
    which we say that \emph{realizes} it,
    such that $\hat{\sgn}_t(S) = \hat{\sigma}_t$, and the number of pairs of connected vertices in the \emph{inactive} components of $G[\mathcal{T}_t] - S$
    is equal to $d$, with $d \le \hat{d} \le (1+\varepsilon) \cdot d$.

    For every node $t \in V(\mathcal{T})$ the algorithm stores a table
    $\DP_t[\cdot, \cdot]$ where, for $\hat{\sigma}_t \in \hat{\Sigma}_t$ and $\hat{d} \in [0, (1+\varepsilon) \cdot \binom{n}{2}]$,
    $\DP_t[\hat{\sigma}_t, \hat{d}]$ is a boolean variable indicating whether there exists an approximate solution $(\hat{\sigma}_t, \hat{d})$.
    Recall that the height $h_t$ of a node $t$ of the decomposition is the largest distance from the node to a leaf rooted in its sub-tree.
    Leaves are at height $0$ and the root is at height $h$.
    In order to show that the algorithm indeed produces a $(1+\varepsilon)$-approximate solution, we want to maintain the following invariant:
    \begin{itemize}
        \item For each node $t$, if there exists an exact solution $(\sigma_t = (Z, k, P, c), d)$,
        then there exists a pseudo-signature $\hat{\sigma}_t = (Z, k, P, \hat{c}) \in \hat{\Sigma}_t$ compatible with $\sigma_t$ such that
        for all $p \in P$, $c(p) \le \hat{c}(p) \le (1+\delta)^{h_t} \cdot c(p)$ and a value $\hat{d}$ with $d \le \hat{d} \le (1+\varepsilon) \cdot d$,
        such that $\DP_t[\hat{\sigma}_t, \hat{d}] = \true$.

        \item For each node $t$, if $\DP_t[\hat{\sigma}_t, \hat{d}] = \true$, where $\hat{\sigma}_t = (Z, k, P, \hat{c})$,
        then there exists an exact solution $(\sigma_t = (Z, k, P, c), d)$ where $\sigma_t$ is compatible with $\hat{\sigma}_t$,
        $d \le \hat{d}$,
        and for all $p \in P$, $c(p) \le \hat{c}(p)$.
    \end{itemize}
    On a high-level, our approximation algorithm slightly overestimates both the size of the active components (by a factor of $(1+\delta)^{h_t}$)
    and the number of pairs of connected vertices in inactive components of the graph (by a factor of $(1+\varepsilon)$):
    on the one hand, if the approximate algorithm computes some solution, an exact solution that is at least that good exists;
    on the other hand, if an exact solution exists, the algorithm will compute one that satisfies the approximation guarantees.
    We now proceed to describe how to populate the DP tables, as well as establish this invariant by induction.

    \proofsubparagraph{Leaf node.}
    If $t$ is a leaf node, then $X_t = \varnothing$ and $\hat{\Sigma}_t = \{ (\varnothing, 0, \{ \{ \varnothing\} \}, [\{ \varnothing\} \mapsto 0] ) \}$.
    In that case, it follows that
    \[
        \DP_t[\hat{\sigma}_t, \hat{d}] =
            \begin{cases}
                \true   &\text{if $\hat{\sigma}_t \in \hat{\Sigma}_t$ and $\hat{d} = 0$,}\\
                \false  &\text{otherwise}.
            \end{cases}
    \]
    Notice that the invariant trivially holds for Leaf nodes.

    \proofsubparagraph{Introduce node.}
    Suppose $t$ is an introduce node with child $t'$ such that $X_t = X_{t'} \cup \{ v \}$ for some $v \notin X_{t'}$.
    Let $\hat{\sigma}_t = (Z, k, P, \hat{c}) \in \hat{\Sigma}_t$ and consider two cases.
    If $v \notin Z$, then we set
    $\DP_t[\hat{\sigma}_t, \hat{d}] = \DP_{t'}[\hat{\sigma}_t, \hat{d}]$ for all $\hat{d} \in [0, (1+\varepsilon) \cdot \binom{n}{2}]$.
    Otherwise it holds that $v \in Z$ and we set
    \[
        \DP_t[\hat{\sigma}_t, \hat{d}] = \bigvee_{\hat{\sigma}_{t'} \in \mathcal{S}^I_{\hat{\sigma}_t}} \DP_{t'}[\hat{\sigma}_{t'}, \hat{d}],
    \]
    where $\mathcal{S}^I_{\hat{\sigma}_t} \subseteq \hat{\Sigma}_{t'}$ such that
    $\hat{\sigma}_{t'} = (Z', k', P', \hat{c'}) \in \mathcal{S}^I_{\hat{\sigma}_t}$ if
    \begin{romanenumerate}
        \item $Z = Z' \cup \{ v \}$,
        \item $k = k'$,
        \item $P'$ is a partition of $Z'$ such that merging all its blocks that contain vertices incident with $v$
        and adding $v$ to the latter results to the partition $P$ of $Z$,
        that is, formally we have $P = P'_{\uparrow\{v\}} \sqcup \{ \{ v \} \cup \setdef{u \in Z'}{u \in N_G(v)} \}_{\uparrow Z}$,
        \item $\hat{c'}$ is a cost function on the blocks of $P'$ such that,
        for $\mathcal{P}' = \setdef{p' \in P'}{p' \cap N_G(v) \neq \varnothing}$,%
        \footnote{That is, $\mathcal{P}'$ contains all the blocks of $P'$ that contain vertices incident with $v$.}
        we have that
        (i) $\hat{c}(p) = \hat{c'}(p)$ for all $p \in P' \setminus \mathcal{P}'$,
        and (ii) $\hat{c}(p_v)$ is equal to the smallest value in $B$ that upper-bounds
        $\sum_{p' \in \mathcal{P}'} \hat{c'}(p')$, where $p_v \in P$ such that $v \in p_v$.
    \end{romanenumerate}

    \begin{lemma}\label{lemma:fpt-as:introduce}
        Let $t$ be an introduce node, with $t'$ denoting its child node.
        Assume that the invariant holds for node $t'$.
        Then it also holds for node $t$.
    \end{lemma}

    \begin{nestedproof}
        Suppose $t$ is an introduce node with child $t'$ such that $X_t = X_{t'} \cup \{ v \}$ for some $v \notin X_{t'}$ and $h_t = h_{t'}+1$.
        By induction, for an exact solution $(\sigma_{t'} = (Z',k',P',c'), d')$ we have calculated an approximate solution
        $(\hat{\sigma}_{t'} = (Z',k',P',\hat{c'}), \hat{d'})$
        with $\hat{c'}(p') \le (1+\delta)^{h_{t'}} \cdot c'(p')$ for all $p' \in P'$ and $\hat{d'} \le (1+\varepsilon) \cdot d'$.

        Let $(\sigma_t = (Z,k,P,c), d)$ denote an exact solution for node $t$, and $S \subseteq X^{\downarrow}_t$ a set that realizes it.
        If $v \notin Z$, it follows that $(\sigma_{t'},d)$ is an exact solution for node $t'$, where $\sigma_{t'} = \sigma_t = \sgn_{t'}(S)$,
        and by induction it easily follows that $(\hat{\sigma}_{t'}, \hat{d})$ is
        an approximate solution for node $t$ that respects the approximation guarantees.
        Alternatively, it holds that $v \in Z$, in which case it follows that there exists an exact solution $(\sigma_{t'}, d)$ in node $t'$
        realized by $S \setminus \{ v \}$.
        By induction, we have calculated the approximate solution $(\hat{\sigma}_{t'}, \hat{d})$,
        and let $\hat{\sigma}_t = (Z,k,P,\hat{c}) \in \hat{\Sigma}_t$ such that $\hat{\sigma}_{t'} \in \mathcal{S}^I_{\hat{\sigma}_t}$.
        We argue that $(\hat{\sigma}_t, \hat{d})$ is an approximate solution that respects the guarantees.
        To see this, it suffices to argue about the size of the active component that contains $v$.
        Let $p_v \in P$ with $v \in p_v$, and let $p'_1, \ldots, p'_w \in P'$ denote the blocks of $P'$ that contain vertices belonging to $N_G(v)$.
        It holds that
        \begin{align*}
            \hat{c}(p_v) &= (1+\delta)^{\lceil \log_{(1+\delta)} (\hat{c'}(p'_1) + \ldots + \hat{c'}(p'_w)) \rceil}\\
            &\le (1+\delta) \cdot (\hat{c'}(p'_1) + \ldots + \hat{c'}(p'_w))\\
            &\le (1+\delta) \cdot (1+\delta)^{h_{t'}} \cdot (c'(p'_1) + \ldots + c'(p'_w))\\
            &= (1+\delta)^{h_t} \cdot c(p_v),
        \end{align*}
        where $\log_{(1+\delta)} (\cdot)$ denotes the logarithm base $(1+\delta)$, and the second inequality is due to the induction hypothesis.

        As for the second direction of the invariant, it is straightforward to argue about its validity by induction.
    \end{nestedproof}

    \proofsubparagraph{Forget node.}
    Suppose $t$ is a forget node with child $t'$ such that $X_t = X_{t'} \setminus \{ v \}$ for some $v \in X_{t'}$.
    For all $\hat{\sigma}_t = (Z, k, P, \hat{c}) \in \hat{\Sigma}_t$ and $\hat{d} \in [0, (1+\varepsilon) \cdot \binom{n}{2}]$ we set
    \[
        \DP_t[\hat{\sigma}_t, \hat{d}] = \bigvee_{\hat{\sigma}_{t'} \in \mathcal{S}^F_{\hat{\sigma}_t}} \DP_{t'}[\hat{\sigma}_{t'}, \hat{d} - \hat{d'}(\hat{\sigma}_{t'})],
    \]
    where $\mathcal{S}^F_{\hat{\sigma}_t} \subseteq \hat{\Sigma}_{t'}$ such that
    $\hat{\sigma}_{t'} = (Z', k', P', \hat{c'}) \in \mathcal{S}^F_{\hat{\sigma}_t}$ if
    \begin{romanenumerate}
        \item $Z = Z' \setminus \{ v \}$,
        \item if $v \in Z'$ then $k = k'$, otherwise $k = k'+1$,
        \item $P = P'_{\downarrow Z}$,
        \item for all $p' \in P'$ such that $v \notin p'$, we have $\hat{c}(p') = \hat{c'}(p')$,

        \item if $\{ v \} \in P'$ then $\hat{d'}(\hat{\sigma}_{t'}) = \binom{\hat{c'}(\{ v \}) + 1}{2}$,
        otherwise $\hat{d'}(\hat{\sigma}_{t'}) = 0$ and, if $v \in Z'$, let $p'_v \in P'$ with $v \in p'_v$,
        in which case let $\hat{c}(p'_v \setminus \{ v \})$ be equal to the minimum element of $B$ that upper-bounds $\hat{c'}(p'_v) + 1$.
    \end{romanenumerate}
    Intuitively, if $v$ is the last vertex of an active component we accumulate the number of pairs of connected vertices in that component
    to $\hat{d}$, otherwise we increase the approximate size of said active component.

    \begin{lemma}\label{lemma:fpt-as:forget}
        Let $t$ be a forget node, with $t'$ denoting its child node.
        Assume that the invariant holds for node $t'$.
        Then it also holds for node $t$.
    \end{lemma}

    \begin{nestedproof}
        Suppose $t$ is a forget node with child $t'$ such that $X_t = X_{t'} \setminus \{ v \}$ for some $v \in X_{t'}$ and $h_t = h_{t'}+1$.
        By induction, for an exact solution $(\sigma_{t'} = (Z',k',P',c'), d')$ we have calculated an approximate solution
        $(\hat{\sigma}_{t'} = (Z',k',P',\hat{c'}), \hat{d'})$
        with $\hat{c'}(p') \le (1+\delta)^{h_{t'}} \cdot c'(p')$ for all $p' \in P'$ and $\hat{d'} \le (1+\varepsilon) \cdot d'$.

        Let $(\sigma_t = (Z,k,P,c), d)$ denote an exact solution for node $t$, and $S \subseteq X^{\downarrow}_t$ a set that realizes it.
        We consider different cases.

        First, if $v \notin S$, it follows that $(\sigma_{t'} = (Z,k-1,P,c),d)$ is an exact solution for node $t'$,
        and by induction it easily follows that $(\hat{\sigma}_{t} = (Z,k,P,\hat{c}), \hat{d})$ is
        an approximate solution for node $t$ that respects the approximation guarantees.

        Alternatively, it holds that $v \in S$, in which case it follows that there exists an exact solution $(\sigma'_t = (Z',k',P',c'), d')$ in node $t'$
        realized by $S$.
        By induction, we have calculated the approximate solution $(\hat{\sigma}_{t'} = (Z',k',P',\hat{c'}), \hat{d'})$
        and let $\hat{\sigma}_t = (Z,k,P,\hat{c}) \in \hat{\Sigma}_t$ such that $\hat{\sigma}_{t'} \in \mathcal{S}^F_{\hat{\sigma}_t}$.
        Let $p'_v \in P'$ with $v \in p'_v$.

        Consider the case where $\{ v \} \notin P'$, that is, $v$ is not the last vertex of an active component.
        We argue that the approximate solution $(\hat{\sigma}_t, \hat{d'})$ satisfies the approximation guarantees.
        It suffices to argue about the approximation guarantee on the size of the active component $p'_v \setminus \{ v \}$;
        indeed, it holds that
        \begin{align*}
            \hat{c}(p'_v \setminus \{ v \}) &= (1+\delta)^{\lceil \log_{(1+\delta)} (\hat{c'}(p'_v) + 1) \rceil}\\
            &\le (1+\delta) \cdot (\hat{c'}(p'_v) + 1)\\
            &\le (1+\delta) \cdot (1+\delta)^{h_{t'}} \cdot (c'(p'_v) + 1)\\
            &= (1+\delta)^{h_t} \cdot c(p'_v \setminus \{v \}),
        \end{align*}
        where $\log_{(1+\delta)} (\cdot)$ denotes the logarithm base $(1+\delta)$, and the second inequality is due to the induction hypothesis.

        Lastly, consider the case where $\{ v \} \in P'$, that is, $v$ is the last vertex of an active component.
        We argue that the approximate solution $(\hat{\sigma}_t, \hat{d'} + \hat{d'}(\hat{\sigma}_{t'}))$ satisfies the approximation guarantees,
        where $\hat{d'}(\hat{\sigma}_{t'}) = \binom{\hat{c'}(\{ v \}) + 1}{2}$.
        In suffices to argue about the approximation guarantee on the number of pairs of connected vertices in the inactive components.
        Let $C$ denote the active component of $G[\mathcal{T}_{t'}] - S$ that contains $v$.
        Notice that if $|V(C)| = 1$, then $C = \{ v \}$ and $c'(\{ v \}) = \hat{c'}(\{ v \}) = 0$, thus $\hat{d'}(\hat{\sigma}_{t'}) = 0$.
        Alternatively, it holds that
        \begin{align*}
            \hat{c'}(p_v) + 1 &\le (1+\delta)^{h_{t'}} \cdot c'(p_v) + 1\\
            &\le (1+\delta)^{h_{t'}} \cdot (c'(p_v) + 1)\\
            &\le (1+\varepsilon') \cdot (c'(p_v) + 1)\\
            &= (1+\varepsilon') \cdot |V(C)|,
        \end{align*}
        where in the second inequality we use the fact that we have set $\delta$ to a value such that
        $(1+\delta)^h \le e^{\delta h} = e^{\varepsilon'/2} \le (1+\varepsilon')$
        for small enough $\varepsilon'$.%
        \footnote{It suffices to assume without loss of generality $\varepsilon' < 1/4$.}
        It follows that the induction hypothesis, along with \Cref{lemma:fpt-as:varepsilon} in the second case,
        implies that indeed $\hat{d'} + \hat{d'}(\hat{\sigma}_{t'})$ is a $(1+\varepsilon)$-approximation of $d$.

        As for the second direction of the invariant, it is straightforward to argue about its validity by induction.
    \end{nestedproof}

    \proofsubparagraph{Join node.}
    Finally, suppose that $t$ is a join node with children $t_1, t_2$ such that $X_t = X_{t_1} = X_{t_2}$.
    For all $\hat{\sigma}_t = (Z, k, P, \hat{c}) \in \hat{\Sigma}_t$ and $\hat{d} \in [0, (1+\varepsilon) \cdot \binom{n}{2}]$ we set
    \[
        \DP_t[\hat{\sigma}_t, \hat{d}] =
        \bigvee_{\genfrac{}{}{0pt}{}{\hat{d}_1 + \hat{d}_2 = \hat{d},}{(\hat{\sigma}_{t_1}, \hat{\sigma}_{t_2}) \in \mathcal{R}^U_{\hat{\sigma}_{t}}}}
        (   \DP_{t_1}[\hat{\sigma}_{t_1}, \hat{d}_1] \land
            \DP_{t_2}[\hat{\sigma}_{t_2}, \hat{d}_2] ),
    \]
    where $\mathcal{R}^U_{\hat{\sigma}_{t}} \subseteq \hat{\Sigma}_{t_1} \times \hat{\Sigma}_{t_2}$ such that
    $(\hat{\sigma}_{t_1}, \hat{\sigma}_{t_2}) = ((Z_1, k_1, P_1, \hat{c}_1), \, (Z_2, k_2, P_2, \hat{c}_2)) \in \mathcal{R}^U_{\hat{\sigma}_{t}}$ if
    \begin{romanenumerate}
        \item $Z = Z_1 = Z_2$,
        \item $k = k_1 + k_2$,
        \item $P = P_1 \sqcup P_2$,
        \item for every $p \in P$, $\hat{c}(p)$ is equal to the minimum element of $B$ that upper-bounds
        \[
            \sum_{p_1 \subseteq p \land p_1 \in P_1} \hat{c}_1(p_1) + \sum_{p_2 \subseteq p \land p_2 \in P_2} \hat{c}_2(p_2).
        \]
    \end{romanenumerate}

    \begin{lemma}\label{lemma:fpt-as:join}
        Let $t$ be a join node, with $t_1$ and $t_2$ denoting its children.
        Assume that the invariant holds for nodes $t_1$ and $t_2$.
        Then it also holds for node $t$.
    \end{lemma}

    \begin{nestedproof}
        Suppose that $t$ is a join node with children $t_1, t_2$ such that $X_t = X_{t_1} = X_{t_2}$ and $h_{t_1}, h_{t_2} \le h_t - 1$.

        Let $(\sigma_t = (Z,k,P,c), d)$ denote an exact solution for node $t$ and $S \subseteq X^{\downarrow}_t$ a set that realizes it.
        Let for $j \in \{ 1,2 \}$, $\sgn_{t_j} (S) = \sigma_j = (Z,k_j,P_j,c_j)$,
        and $(\sigma_j, d_j)$ denote the exact solution at node $t_j$ that is realized by $S \cap X^{\downarrow}_{t_j}$.
        By induction, we have calculated, for $j \in \{ 1,2 \}$, approximate solutions $(\hat{\sigma}_{t_1} = (Z,k_j,P_j,\hat{c}_j), \hat{d}_j)$
        with $\hat{c}_j(p) \le (1+\delta)^{h_{t_j}} \cdot c_j(p)$ for all $p \in P_j$ and $\hat{d}_j \le (1+\varepsilon) \cdot d_j$.
        Consider the approximate solution $(\hat{\sigma}_t, \hat{d}_t)$ where
        $(\hat{\sigma}_1, \hat{\sigma}_2) \in \mathcal{R}^U_{\hat{\sigma}_t}$ and $\hat{d} = \hat{d}_1 + \hat{d}_2$.
        We argue that it respects the approximation guarantees.
        Indeed, it holds that $\hat{d}_1 + \hat{d}_2 \le (1+\varepsilon) \cdot (d_1+d_2) = (1+\varepsilon) \cdot d$.
        Moreover, for $p \in P$ it holds that
        \begin{align*}
            \hat{c}(p) &= (1+\delta)^{\big\lceil \log_{(1+\delta)}
                \big( \sum_{p_1 \subseteq p \land p_1 \in P_1} \hat{c}_1(p_1) + \sum_{p_2 \subseteq p \land p_2 \in P_2} \hat{c}_2(p_2) \big) \big\rceil}\\
            &\le (1+\delta) \cdot \Big( \sum_{p_1 \subseteq p \land p_1 \in P_1} \hat{c}_1(p_1) + \sum_{p_2 \subseteq p \land p_2 \in P_2} \hat{c}_2(p_2) \Big)\\
            &\le (1+\delta) \cdot (1+\delta)^{h_t - 1} \cdot
                \Big( \sum_{p_1 \subseteq p \land p_1 \in P_1} c_1(p_1) + \sum_{p_2 \subseteq p \land p_2 \in P_2} c_2(p_2) \Big)\\
            &= (1+\delta)^{h_t} \cdot c(p),
        \end{align*}
        where $\log_{(1+\delta)} (\cdot)$ denotes the logarithm base $(1+\delta)$, and the second inequality is due to the induction hypothesis.

        As for the second direction of the invariant, it is straightforward to argue about its validity by induction.
    \end{nestedproof}

    \proofsubparagraph{Wrapping up.}
    The correctness of the algorithm follows by \Cref{lemma:fpt-as:introduce,lemma:fpt-as:forget,lemma:fpt-as:join}.
    Using back-tracking we can compute a set that realizes a given approximate solution with polynomial overhead.
    It remains to argue about the running time.
    Notice that $|B| = \bO \big( \frac{\log n}{\log (1+\delta)} \big) =
    \bO \big( \frac{\log n}{\delta} \big) =
    \bO \big( \frac{\tw \log^2 n}{\varepsilon'} \big)$,
    where we used the fact that $e^{\delta/2} \le 1+\delta$ for all $\delta < 1/2$.
    It follows that the running time of the presented algorithm is $2^{\bO(\tw)} \tw^{\bO(\tw)} |B|^{\bO(\tw)} n^{\bO(1)}
    = (\log n / \varepsilon)^{\bO(\tw)} \tw^{\bO(\tw)} n^{\bO(1)}$.
    To achieve the running time bound of $(\tw / \varepsilon)^{\bO(\tw)} n^{\bO(1)}$ we use a well-known win/win argument:
    if $\tw \le \sqrt{\log n}$ then $(\log n)^{\bO(\tw)} =
    (\log n)^{\bO(\sqrt{\log n})} =
    n^{o(1)}$;
    alternatively, $\log n \le \tw^2$ thus $(\log n)^{\bO(\tw)} = \tw^{\bO(\tw)}$.
\end{proof}

\begin{proofsketch}
    Here we describe the main idea behind our algorithm.
    On a high level, we aim to develop a DP which, while traversing the tree decomposition,
    keeps track of the sizes of any \emph{active} components (those whose vertices intersect the bag),
    while for the rest of the components (i.e., the \emph{inactive} ones)
    there exists a variable on which we account for their number of pairs of connected vertices.
    To this end, assuming that the exact size of an active component is $c$, our DP stores a value $\hat{c} \le (1+\varepsilon') \cdot c$
    where $\varepsilon'$ is such that $\binom{\hat{c}}{2} \le (1+\varepsilon) \cdot \binom{c}{2}$.
    Notice that for this to hold, $c=1$ implies that $\hat{c}=1$.
    Consequently, it suffices to present a dynamic program which correctly stores the size of any singleton active component,
    while for the rest of active components it allows for $(1 + \varepsilon')$-approximate values on their sizes.
    In that case, the number of pairs of connected vertices accounted for every connected component is a $(1+\varepsilon)$-approximation,
    and since we sum over those, the final value has a $(1+\varepsilon)$-approximation ratio as well.
    Our algorithm is thus a DP that does exactly as required.
\end{proofsketch}

\section{Kernelization}\label{sec:kernel}

Given that, as we have shown in \cref{thm:CNC:vi}, {\CNC} is FPT parameterized by the vertex integrity of the input graph,
a natural question arising is whether one can develop a polynomial kernel under this parameterization.
In this section we prove that this cannot be the case under standard assumptions,
even for the much more restricted parameterization by vertex cover number.

\begin{theorem}\label{thm:kernel}
    {\CNC} does not admit a polynomial kernel when parameterized by the vertex cover number
    of the input graph, unless $\mathsf{NP} \subseteq \mathsf{coNP} \slash \mathsf{poly}$.
\end{theorem}

\begin{proof}
    We present a polynomial parameter transformation reducing from \kMC.
    In the latter, we are given a graph $G = (V,E)$ and a partition of $V$ into $k$ independent sets (also called color classes)
    $V_1, \ldots, V_k$, each of size $n$,
    and we are asked to determine whether $G$ contains a $k$-clique.
    It is known that {\kMC} parameterized by $k \log n$ does not admit a polynomial kernel
    unless $\mathsf{NP} \subseteq \mathsf{coNP} \slash \mathsf{poly}$~\cite{algorithmica/HermelinKSWW15}.

    \proofsubparagraph{Construction.}
    Consider an instance $(G,k)$ of \kMC, where $V_i = \setdef{v^i_j}{j \in [n]}$ for all $i \in [k]$.
    For each color class $V_i$, we assign to each of its vertices a unique (among the vertices of $V_i$) bit-string of length $\log n$.
    For all $w \in [\log n]$,
    let $b_w \colon V \to \{ 0,1 \}$ denote the $w$-th bit in the bit-string of $v^i_j \in V_i$. 
    Construct the graph $G'$ as follows.
    \begin{itemize}
        \item We first introduce a clique on vertex set
        $K = \setdef{g^{i,w}_z}{i \in [k], \, w \in [\log n], \, z \in \{0,1\}}$,
        which we refer to as \emph{core vertices}.
        We say that a core vertex $g^{i,w}_z$ \emph{corresponds} to color class $i$,
        with $K_i$ being composed of all core vertices corresponding to color class $i$, for $i \in [k]$.
        Furthermore, we say that a vertex $v \in V_i$ of $G$ is \emph{encoded} by the vertices $\setdef{g^{i,w}_{b_w(v)}}{w \in [\log n]}$.

        \item Introduce a clique on the vertex set $C = \setdef{c_i}{i \in [k \log n + 1]}$,
        and add edges such that every vertex of $C$ is adjacent to all the core vertices.

        \item For every edge $e = \{ u_1, u_2 \} \in E(G)$,
        where $u_1 \in V_{i_1}$ and $u_2 \in V_{i_2}$,
        introduce an \emph{adjacency vertex} $h_e$ which is incident with all vertices
        encoding its endpoints, that is, with vertices $g^{i_1, w}_{b_w(u_1)}$ and $g^{i_2, w}_{b_w(u_2)}$
        for all $w \in [\log n]$.
        Let $H_{i_1,i_2}$ denote the set of all adjacency vertices due to edges in $G$ between vertices in
        $V_{i_1}$ and $V_{i_2}$. 

        \item For all $\{ i_1, i_2 \} \in \binom{[k]}{2}$ (that is, for all pairs of different color classes),
        and for all $w_1, w_2 \in [\log n]$ and $z_1, z_2 \in \{0,1\}$,
        introduce an independent set of size $A=|E|+1$,
        each vertex of which is incident with $g^{i_1,w_1}_{z_1}$ and $g^{i_2,w_2}_{z_2}$.
        We refer to the vertices added in this step of the construction as \emph{dummy vertices}.
    \end{itemize}
    This concludes the construction of the graph $G'$.
    Notice that $G' - (K \cup C)$ is an independent set,
    consequently the vertex cover number of $G'$ is at most $3k \log n + 1$.
    We will show that $(G', k', x)$ is an equivalent instance of \CNC,
    where $k' = k \log n$ and $x = \binom{|V(G')| - k' - \binom{k}{2} - A\binom{k}{2} \log^2 n}{2}$.

    For the forward direction, consider a function $s \colon [k] \to [n]$ such that
    $\mathcal{V} = \setdef{v^i_{s(i)}}{i \in [k]}$ is a $k$-clique in $G$.
    In that case, let $S = \setdef{g^{i,w}_{b_w(v^i_{s(i)})}}{i \in [k], \, w \in [\log n]}$
    be a set of size $k'$.
    We will prove that $G'-S$ has at most $x$ pairs of connected vertices.
    First, notice that for every pair of core vertices belonging to $S$ that correspond to different color classes,
    their removal results in an independent set of dummy vertices of size $A$ in $G'-S$.
    Since there are $\binom{k}{2} \log^2 n$ such pairs, $G'-S$ contains $A \binom{k}{2} \log^2 n$ isolated dummy vertices.
    Furthermore, we argue that the adjacency vertex of any edge in $G[\mathcal{V}]$ is isolated in $G'-S$.
    To see this, consider the adjacency vertex $h_e \in V(G')$ where $e = \{ v^{i_1}_{s(i_1)}, v^{i_2}_{s(i_2)} \} \in E(G)$.
    It holds that
    $N_{G'} (h_e) = \setdef{g^{{i_1},w}_{b_w (v^{i_1}_{s(i_1)})}, \, g^{{i_2},w}_{b_w (v^{i_2}_{s(i_2)})}}{w \in [\log n]} \subseteq S$,
    thus $h_e$ is indeed an isolated vertex in $G'-S$.
    Since $\mathcal{V}$ induces a $k$-clique, there are $\binom{k}{2}$ such isolated adjacency vertices.
    Finally, $S$ itself is of size $k'$.
    Consequently, the number of pairs of connected vertices in $G'-S$ is at most the number of pairs of non-isolated vertices in the graph,
    which is at most $\binom{|V(G')| - A \binom{k}{2} \log^2 n - \binom{k}{2} - k'}{2} = x$.

    For the opposite direction, let $S \subseteq V(G')$ be of size at most $k'$ such that $G'-S$ has at most $x$ pairs of connected vertices.
    Notice that $|C| > k'$, consequently there exists a vertex $c \in C \setminus S$.
    Since $N_G(c) \supseteq K$ and $C$ is a clique, it follows that all vertices of $(K \cup C) \setminus S$ are in the same connected component of $G'-S$;
    let this component be denoted by $\BigComponent \in \cc(G'-S)$.
    Furthermore, since for any vertex $v \in V(G') \setminus (K \cup C)$ it holds that $N_{G}(v) \subseteq K$,
    any such vertex either belongs to $\BigComponent$ or is isolated in $G'-S$.

    \begin{claim}\label{claim:kernel:structure}
        We have $|S \cap K_i| = \log n$, for all $i \in [k]$.
    \end{claim}


    \begin{claimproof}
        To prove the claim we argue about the number of isolated dummy vertices in $G'-S$.
        First, notice that
        $(n-\varepsilon)(n+\varepsilon) = n^2 - \varepsilon^2 < n^2$ for all $\varepsilon>0$.
        Consequently, the number of isolated dummy vertices in $G'-S$ is maximized when
        $|S \cap K_i| = \log n$ for all $i \in [k]$, in which case the number of isolated dummy vertices is $A \binom{k}{2} \log^2 n$.

        Assume that the claim is false. In that case, the number of isolated dummy vertices in $G'-S$
        is at most $A \parens*{\binom{k}{2} \log^2 n - 1}$;
        any other dummy vertex in $G'-S$ must belong to $\BigComponent$, which also contains all vertices of $(K \cup C) \setminus S$.
        Consequently, the size of $\BigComponent \in \cc(G'-S)$ is at least $|V(G')| - |E| - k' -  A \parens*{\binom{k}{2} \log^2 n - 1}$.
        In that case
        \begin{align*}
            |V(G')| - |E| - k' -  A \parens*{\binom{k}{2} \log^2 n - 1}
            &> |V(G')| - A \binom{k}{2} \log^2 n - \binom{k}{2} - k' \iff\\
            - |E| + A &> - \binom{k}{2} \impliedby\\
            A &= |E|+1,
        \end{align*}
        thereby yielding a contradiction as $G'-S$ has more than $x$ pairs of connected vertices.
    \end{claimproof}


    Recall that $G'-S$ is composed of a connected component $\BigComponent$ containing all non-isolated vertices, as well as some isolated vertices.
    By \cref{claim:kernel:structure} it holds that $|S| = k'$, thus for $G'-S$ to have at most $x$ pairs of connected vertices the
    number of its isolated vertices must be at least $\binom{k}{2} + A\binom{k}{2} \log^2 n$.
    Due to \cref{claim:kernel:structure} it follows that it has exactly $A\binom{k}{2} \log^2 n$ isolated dummy vertices,
    while none of the remaining vertices of $K \cup C$ can be isolated.
    Consequently, the deletion of $S$ isolates at least $\binom{k}{2}$ adjacency vertices.

    We argue that no two isolated adjacency vertices $h_{e_1},h_{e_2}$ belong to the same set $H_{i_1,i_2}$.
    Assume that this is the case, and let $h_{e_1},h_{e_2} \in H_{i_1,i_2}$.
    For $u \in \{ h_{e_1},h_{e_2} \}$ it holds that $|N_{G'}(u) \cap K_{i_1}| = |N_{G'}(u) \cap K_{i_2}| = \log n$,
    while $N_{G'}(h_{e_1}) \neq N_{G'}(h_{e_2})$.
    Due to \cref{claim:kernel:structure}, this leads to a contradiction.
    Consequently, there is a exactly one isolated adjacency vertex in $G'-S$ belonging to $H_{i_1,i_2}$ for all $\{i_1,i_2\} \in \binom{[k]}{2}$.

    Now consider one such isolated adjacency vertex $h_e \in H_{i_1,i_2}$.
    Notice that $|N_{G'}(h_e) \cap \{ g^{p,w}_{0}, g^{p,w}_{1} \}| = 1$ for all $w \in [\log n]$ and $p \in \{i_1,i_2\}$.
    Since $S \supseteq N_{G'}(h_e)$, and this holds for all isolated adjacency vertices,
    due to \cref{claim:kernel:structure}
    it follows that $|S \cap \{ g^{i,w}_{0}, g^{i,w}_{1} \}| = 1$ for all $w \in [\log n]$ and $i \in [k]$.

    Let $s \colon [k] \to [n]$ such that $v^i_{s(i)}$ is encoded by the vertices $S \cap K_i$, for all $i \in [k]$.
    We claim that $\mathcal{V} = \setdef{v^i_{s(i)}}{i \in [k]}$ induces a $k$-clique in $G$.
    Consider $\{i_1,i_2\} \in \binom{[k]}{2}$.
    Notice that there exists an isolated adjacency vertex $h_e$ in $G'-S$ belonging to $H_{i_1,i_2}$;
    the neighborhood of $h_e$ is exactly the vertices encoding its two endpoints in $G$,
    thus $\{v^{i_1}_{s(i_1)}, v^{i_2}_{s(i_2)}\} \in E(G)$.
\end{proof}

%




%

\section{Conclusion}

In this paper we have thoroughly studied {\CNC} under the perspective of parameterized complexity,
and presented a plethora of results, mostly taking into account the structure of the input graph.
As a direction of future work, it is currently unknown whether the problem is FPT when parameterized by
$k+\td$, where $\td$ denotes the tree-depth of the input graph, or when parameterized by the cluster vertex deletion number
of the input graph. Another interesting direction would be to (dis)prove the optimality of the $n^{\bO(\tw)}$ or the
$n^{\bO(2^\cw)}$ algorithm; especially for the latter, we note that there are only a handful of
natural problems for which such a running time is known to be optimal under the ETH~\cite{algorithmica/AboulkerBKS23,talg/FominGLSZ19,mst/JaffkeLL24}.
Given the similarity of the construction of \cref{thm:CNC:fes} and the algorithm of \cref{thm:CNC:ml} with
the corresponding results for {\VI}~\cite{mfcs/HanakaLVY24}, a natural question is whether we can
obtain similar results for other optimization functions that take into account the sizes of the
components of the graph remaining after the vertex deletions. Some of our results seem to be easily adaptable
to such a more general setting (e.g., for general separately convex functions),
however we do not know whether that is the case for the W[1]-hardness as well.

\bibliography{bibliography}

\end{document}